\documentclass[journal]{IEEEtran}
\usepackage{epsfig,latexsym,amsmath,amssymb,setspace,cite,color,graphicx,algorithm,algorithmic,cases}

\usepackage[caption=false,font=footnotesize]{subfig}
\usepackage[percent]{overpic}
\usepackage{graphicx} 

\usepackage{amsmath}
\usepackage{amsfonts}
\usepackage{amssymb}
\usepackage{dsfont}
\usepackage{bm}
\usepackage{amsthm}
\usepackage{newlfont}
\usepackage{float}
\usepackage{hyperref}
\usepackage{algorithm}
\usepackage{algorithmic}
\usepackage{enumerate}
\usepackage{chngcntr}
\usepackage{mathtools}
\usepackage{breqn}
\usepackage{stmaryrd}
\usepackage{cite}
\usepackage[yyyymmdd]{datetime}

\hypersetup{
    colorlinks=true, 
    linkcolor= blue, 
    citecolor=blue 
}

\begin{document}
\allowdisplaybreaks[1]

\newtheorem{thm}{Theorem} 
\newtheorem{lem}{Lemma}
\newtheorem{prop}{Proposition}
\newtheorem{cor}{Corollary}
\newtheorem{defn}{Definition}
\newcommand{\remarkend}{\IEEEQEDopen}
\newtheorem{remark}{Remark}
\newtheorem{rem}{Remark}
\newtheorem{ex}{Example}
\newtheorem{pro}{Property}

\newenvironment{example}[1][Example]{\begin{trivlist}
\item[\hskip \labelsep {\bfseries #1}]}{\end{trivlist}}

\renewcommand{\qedsymbol}{ \begin{tiny}$\blacksquare$ \end{tiny} }


\renewcommand{\leq}{\leqslant}
\renewcommand{\geq}{\geqslant}


\title { 
Secret-Key Generation in Many-to-One Networks: An Integrated Game-Theoretic and Information-Theoretic Approach}%

\author{\IEEEauthorblockN{R\'emi A. Chou, \textit{Member, IEEE}, and  Aylin Yener, \textit{Fellow, IEEE}}
\thanks{
R\'{e}mi A. Chou is with the Department of Electrical Engineering and Computer Science, Wichita State University, Wichita, KS 67260. Part of this research was conducted when he was a postdoctoral researcher at the Department of  of Electrical Engineering at The Pennsylvania State University, University Park, PA 16802. Aylin Yener is with the Department of Electrical Engineering, The Pennsylvania State University, University Park, PA 16802. Part of this work was presented at the 2017 IEEE International Symposium on Information Theory (ISIT) in \cite{ISIT17a}.
This research has been sponsored by the National Science Foundation Grant CNS-1314719. }
}
\maketitle
\begin{abstract}
This paper considers secret-key generation between several
agents and a base station that observe independent and identically
distributed realizations of correlated random variables.
Each agent wishes to generate the longest possible individual
key with the base station by means of public communication. All
keys must be jointly kept secret from all external entities. 
In this
many-to-one secret-key generation setting, it can be shown that
the agents can take advantage of a collective protocol to increase
the sum-rate of their generated keys. However, when each
agent is only interested in maximizing its own secret-key rate,
agents may be unwilling to participate in a collective protocol.
Furthermore, when such a collective protocol is employed, how
to fairly allocate individual key rates arises as a valid issue.
This paper studies this tension between cooperation and self-interest with
a game-theoretic treatment. The work establishes that cooperation is in
the best interest of all individualistic agents and that there exists individual
secret-key rate allocations that incentivize the agents to follow the
protocol. Additionally, an explicit coding scheme that
achieves such allocations is proposed.
\end{abstract} 

\section{Introduction}

Multiterminal communication settings subject to limited total resources bring about issues pertaining to competition,
and fairness among users. %
Such issues are typically studied by means of game theory; see, for instance,  references \cite{la2004game,gajic2008game,zhu2009constrained}  which deal with the Gaussian multiple access channel, and~\cite{leshem2008cooperative,mathur2006coalitional,berry2011shannon,liu2011game,jorswieck2008complete,larsson2008competition} which deal with interference channels.

In this paper, we study a multiterminal secret-key generation problem that involves selfish users, and propose to study the tension between cooperation and selfishness by means of cooperative game theory, more specifically, coalitional game theory.  We refer to~\cite{peleg2007introduction,osborne1994course,myerson2013game} for an introduction to coalitional game theory, and to~\cite{saad2009coalitional} for a review of some of its applications to telecommunications. %
Our setting can be explained as follows. Each agent wishes to generate an individual key of maximal length with the base station to securely and individually report information, using a one-time pad for instance. 
There are many such agents and a single base station. The generated keys must be jointly kept secret from all external entities. 
We consider a source model for secret-key generation \cite{Ahlswede93,Maurer93}, i.e., the agents and the base station observe independent and identically
distributed (i.i.d.) realizations of correlated random variables (possibly obtained, after appropriate manipulations, from \emph{channel gains measurements}~\cite{wilson2007channel,wallace2010automatic,ye2010information,pierrot2013experimental}), and can communicate over an authenticated public noiseless channel.  
It can be shown that when agents are altruistic, the agents increase the sum of their key lengths by agreeing to participate in a joint protocol, in contrast to operating separately on their own. However, when each agent is  interested in maximizing its own key length only, as we consider, there exists a tension between  cooperation and the sole interest of a given agent. Moreover, assuming that the agents collaborate to maximize the sum of their key lengths, another issue is to determine a fair allocation of individual key lengths,  so that no agent has any incentive to deviate from the protocol. 
The goal of our study is to study this tension between cooperation and selfishness.

Note that when the agents are assumed to be altruistic, and when fairness issues are ignored, the secret-key generation model we consider reduces to the one studied in~\cite{lai2013simultaneously} and is related to multiple-key generation in a network with trusted helpers~\cite{lai2015key,zhangmulti,xu2016simultaneously}. Note also that once the secret-key generation protocol is done, the subsequent transmission to the base station of messages protected by means of a one-time pad with the generated secret keys can be viewed as a noiseless multiple access wiretap channel \cite{tekin2008general}.

Our contributions are three-fold. (i) We formally introduce an integrated game-theoretic and information-theoretic formulation of the problem  in Section \ref{sec:def}. Specifically, we cast the problem as a coalitional game in which the value function is determined under  information-theoretic guarantees, i.e., the value associated with a coalition is computed with no restrictions  on the strategies that the users outside the coalition could adopt. We then derive properties of the defined game and propose rate allocations as candidates for fair solutions in Section \ref{sec:analysis}. (ii)~By adding the constraint that the agents are selfish, we derive a converse using the core of the game we define, which differs from the techniques used for a setting that does not involve selfihness constraints~\cite{lai2013simultaneously,zhangmulti}. (iii)  We provide in Section~\ref{sec:CS} an explicit coding scheme based on polar codes for source coding \cite{Arikan10} and hash functions to implement the solutions proposed in Section~\ref{sec:analysis}. Note that a few explicit coding schemes have been proposed for multiterminal secret-key generation problems~\cite{Nitinawarat10,Ye12,chou2015polar}, however, the coding schemes in these references do not seem to easily apply to our setting. Specifically, the distributed nature of our setting is challenging as each agent must locally generate a key without the knowledge of the source observations of the other agents, and all the generated keys must be collectively secure.

The remainder of the paper is organized as follows. We formally state the problem in Section~\ref{sec:def}. We study in Section~\ref{sec:analysis} the game we have defined in Section \ref{sec:def}. We propose an explicit coding scheme to achieve any point in the core of our game in Section \ref{sec:CS}. We study our model in the case of non-degraded sources in Section \ref{sec:ndg}.  We generalize our model to a setting with multiple clearance levels in Section \ref{secext}. Finally, we provide concluding remarks in Section~\ref{secconcl}.

\section{Problem Statement} \label{sec:def}

We define an auxiliary secret-key generation model with no selfishness constraints  in Section~\ref{secmod}, and provide additional definitions in Section \ref{secmod3}. In Section~\ref{secmod2}, we explain our objective using the model of Section~\ref{secmod} to which selfishness constraints are added, and describe the integrated game-theoretic and information-theoretic problem formulation. %

Notation: For any $a \in \mathbb{N}^*$, define $\llbracket 1, a \rrbracket \triangleq [1,a] \cap \mathbb{N}$. For a given set $\mathcal{S}$, we  let $2^{\mathcal{S}}$  denote the power set of~$\mathcal{S}$. For two probability distributions $p$ and $q$ defined over the same alphabet $\mathcal{X}$, we define the variational distance between $p$ and $q$ as 
$\mathbb{V}(p,q) \triangleq \sum_{x\in \mathcal{X}} |p(x) - q(x)| .$ Finally, $\bigtimes$ denotes the Cartesian product.

\begin{figure} 
\centering

\includegraphics[width=8.5cm]{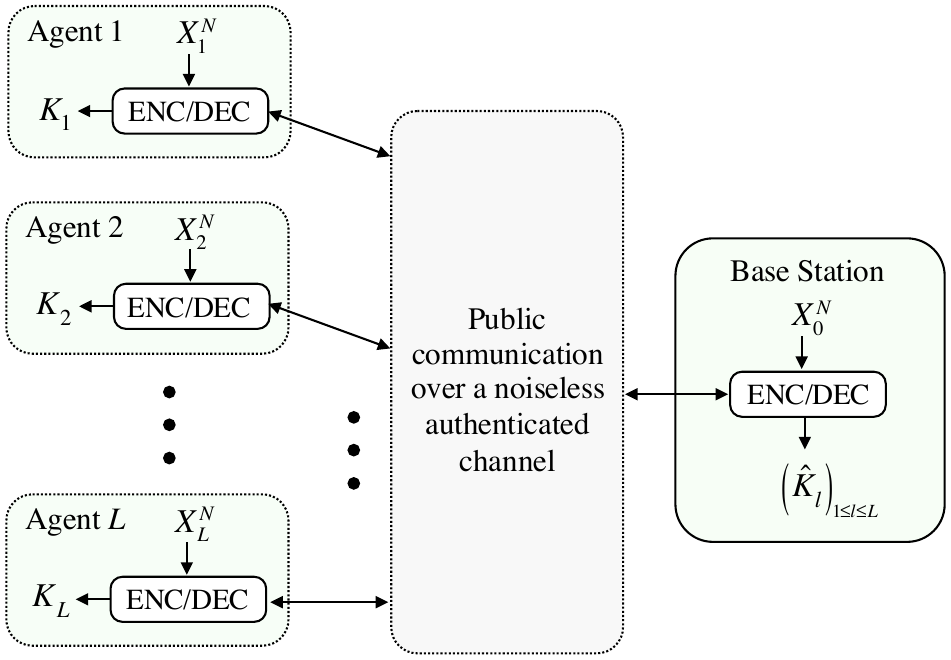}  %
  \caption{Many-to-one secret-key generation setting.}  \label{fig:modellevela}
\end{figure}

\subsection{An auxiliary secret-key generation model (without selfishness constraints)} \label{secmod}
Let $L \in \mathbb{N}^*$ and $\mathcal{L} \triangleq \llbracket 1 ,L \rrbracket$. In the following, we consider $L$ agents represented by the set $\mathcal{L}$ and one base station. %

\subsubsection{Definition of the source model}Define $\mathcal{X}_{\mathcal{L}}$ as the Cartesian product of $L$ finite alphabets $\mathcal{X}_{l}$, $l\in \mathcal{L}$. Consider a discrete memoryless source (DMS) $\left(\mathcal{X}_{\mathcal{L}} \times \mathcal{X}_0 , p_{X_{\mathcal{L}}X_0} \right)$, where $\mathcal{X}_0$ is a finite alphabet and $X_{\mathcal{L}} \triangleq (X_{l})_{l \in \mathcal{L} }$. For $l \in\mathcal{L}$, Agent $l$ observes the component $X_{l}$ of the DMS, and the base station observes the component $X_0$. The source is assumed to follow the following Markov chain: for any $\mathcal{S},\mathcal{T} \subset \mathcal{L}$ such that $\mathcal{S} \cap \mathcal{T} = \emptyset$,
\begin{align} \label{Markov}
X_{\mathcal{S}} - X_0 - X_{\mathcal{T}}.
\end{align}
Note that such a source model has already been considered in~\cite{Csiszar04,Ye12,chou2015polar}. Assuming that all the random variables are binary, an instance of this model is $X_l \triangleq X_0 \oplus B_l$, $\forall l \in \mathcal{L}$, where the $B_l$'s are independent Bernoulli random variables and $\oplus$ is the modulo-two addition. 

The source's statistics are assumed known to all parties, and communication is allowed over an authenticated  noiseless  public channel. 

\subsubsection{Description of the objectives for the agents} The goal of Agent $l\in \mathcal{L}$ is to generate an individual secret-key with the base station. We formalize the definition of a secret-key generation protocol for this setting, which is depicted in Figure~\ref{fig:modellevela}.

\begin{defn} \label{definition_modelg}
For $l \in \mathcal{L}$, let  $\mathcal{K}_l$ be a key alphabet of size $2^{NR_l}$ and define $\mathcal{K}_{\mathcal{L}}$ as the Cartesian product $\bigtimes_{l\in \mathcal{L}} \mathcal{K}_l$. A $( (2^{NR_l})_{l \in \mathcal{L}},N)$ secret-key generation strategy for the agents in  $\mathcal{L}$  is as follows.
\begin{enumerate}
\item The base station observes $X_0^{N}$ and Agent $l$, $l \in \mathcal{L}$, observes $X_l^{N}$.
\item The agents in $\mathcal{L}$ and the base station communicate, possibly interactively, over the public channel. The global public communication between the agents in $\mathcal{L}$ and the base station is denoted by $A \in \mathcal{A}$, for some discrete alphabet $\mathcal{A}$. 
\item Agent $l$, $l \in \mathcal{L}$, computes $K_l(X_l^{N},A) \in \mathcal{K}_l$. 
\item The base station computes $\widehat{K}_l(X_0^{N},A) \in \mathcal{K}_l$, $l \in \mathcal{L}$.
\end{enumerate}
In the following, we use the notation ${K}_{\mathcal{L}} \triangleq  ({K}_l)_{l \in \mathcal{L}}$. 
\end{defn}
\begin{defn} \label{def}
A secret-key rate tuple $(R_l)_{l\in \mathcal{L}}$ is achievable if there exists a sequence of $((2^{NR_l})_{{l\in \mathcal{L}}},N)$ secret-key generation strategies for the agents in  $\mathcal{L}$ such that
\begin{align}
	\lim_{N \to \infty} \mathbb{P}[\widehat{K}_{\mathcal{L}} \neq K_{\mathcal{L}}] &= 0 \text{ (Reliability),} \\
	\lim_{N \to \infty} I\left(K_{\mathcal{L}} ; A \right) &= 0 \text{ (Collective Secrecy),}\label{eqSeca} \\
	\lim_{N \to \infty} \log |\mathcal{K}_{\mathcal{L}} | - H(K_{\mathcal{L}}) &= 0 \text{ (Keys Uniformity)}.\label{eqU1}
\end{align}
\end{defn}
 The secrecy constraint \eqref{eqSeca} ensures that the keys generated by the agents in $\mathcal{L}$ are independent from the public communication.  Note, however, that \eqref{eqSeca} does not mean that the key of a particular agent in $\mathcal{L}$ is secret from the other agents. Moreover, \eqref{eqU1} ensures that the keys generated are almost jointly independent, so that the simultaneous use of the keys by the agents in $\mathcal{L}$ is secure. Note that this setting has been introduced in \cite{lai2013simultaneously}.

Observe that in the presented setting we implicitly assumed that the agents in $\mathcal{L}$ are willing to agree on a common secret key generation protocol. In Section \ref{secmod2}, we study a similar setting but with the additional constraint that the agents are selfish. Before we move to this setting we introduce additional definitions in Section \ref{secmod3}.

\subsection{Additional definitions} \label{secmod3}

We provide additional definitions that will be useful to incorporate selfishness constraints in the model presented in Section \ref{secmod}. These definitions generalize the setting described in Section~\ref{secmod} when security constraints with respect to a subset of agents hold, i.e., the keys generated by a given subset of agents are required to be secret from the rest of the agents. We formalize the definition of a secret-key generation protocol for a group of agents $\mathcal{S} \subseteq \mathcal{L}$ in the following definitions. %

\begin{defn} \label{definition_modelgs}
Let $\mathcal{S} \subseteq \mathcal{L}$. For $i \in \mathcal{S}$, let  $\mathcal{K}_i$ be a key alphabet of size $2^{NR_i}$ and define $\mathcal{K}_{\mathcal{S}}= \bigtimes_{i\in \mathcal{S}} \mathcal{K}_i$. A $( (2^{NR_i})_{i \in \mathcal{S}},N)$ secret-key generation strategy for the coalition of agents $\mathcal{S}$  is as follows.
\begin{enumerate}
\item The base station observes $X_0^{N}$ and Agent $i$, $i \in \mathcal{S}$, observes $X_i^{N}$.
\item The agents in $\mathcal{S}$ and the base station communicate, possibly interactively, over the public channel. The global public communication between the agents in $\mathcal{S}$ and the base station is denoted by $A_{\mathcal{S}} \in \mathcal{A}_{\mathcal{S}}$, for some discrete alphabet $\mathcal{A}_{\mathcal{S}}$. 
\item Agent $i$, $i \in \mathcal{S}$, computes $K_i(X_i^{N},A_{\mathcal{S}}) \in \mathcal{K}_i$. 
\item The base station computes $\widehat{K}_i(X_0^{N},A_{\mathcal{S}}) \in \mathcal{K}_i$, $i \in \mathcal{S}$.
\end{enumerate}
In the following, we use the notation ${K}_{\mathcal{S}} \triangleq  ({K}_i)_{i \in \mathcal{S}}$. 
\end{defn}
\begin{defn} \label{defs}
Let $\mathcal{S} \subseteq \mathcal{L}$. A secret-key rate tuple $(R_i)_{i\in \mathcal{S}}$ is achievable if there exists a sequence of $((2^{NR_i})_{{i\in \mathcal{S}}},N)$ secret-key generation strategies for the coalition of agents $\mathcal{S}$ such that
\begin{align}
	\lim_{N \to \infty} \mathbb{P}[\widehat{K}_{\mathcal{S}} \neq K_{\mathcal{S}}] &= 0 \text{ (Reliability),} \\
	\lim_{N \to \infty} I\left(K_{\mathcal{S}} ; A_{\mathcal{S}}, X_{\mathcal{L} \backslash \mathcal{S}}^{N}  \right) &= 0 \text{ (Collective  Secrecy),}\label{eqSec} \\
	\lim_{N \to \infty} \log |\mathcal{K}_{\mathcal{S}} | - H(K_{\mathcal{S}}) &= 0 \text{ (Key Uniformity)}.\label{eqU}
\end{align}
\end{defn}

 The secrecy constraint~\eqref{eqSec} with respect to the agents outside of $\mathcal{S}$ means
 that the agents in $\mathcal{S}$  follow a protocol for secret-key generation under the information-theoretic constraint that the agent in $\mathcal{L} \backslash \mathcal{S}$ are not assumed to follow \emph{any specific communication strategy}. Note that choosing $\mathcal{S} = \mathcal{L}$ recovers the setting of Section \ref{secmod}. %

\begin{rem} \label{rem}
\eqref{eqSec} and \eqref{eqU} can be combined in only one condition. If 
\begin{align}
\lim_{N \to \infty} N \mathbb{V} \left(p_{K_{\mathcal{S}} A_{\mathcal{S}}X^N_{\mathcal{L} \backslash \mathcal{S}} }, p_{\mathcal{U}_{\mathcal{S}} }p_{A_{\mathcal{S}}X^N_{\mathcal{L} \backslash \mathcal{S}} } \right) =0,
\end{align} then \eqref{eqSec} and \eqref{eqU} hold by \cite[Lemma 1]{Csiszar1996}, \cite[Lemma 2.7]{bookCsizar}, where $p_{\mathcal{U}_{\mathcal{S}}}$ denotes the uniform distribution over $\mathcal{K}_{\mathcal{S}}$.
\end{rem}

\subsection{Secret-key generation with selfish users}  \label{secmod2}
We consider the secret-key generation problem described by Definitions \ref{definition_modelg} and \ref{def}  when the agents are selfish, i.e., they are solely interested in maximizing their own secret-key rate.
The agents can potentially form coalitions to achieve this goal, in the sense that subsets of agents can agree on a collective protocol to follow before the actual secret-key generation protocol occurs.
Note that the model allows the agents to communicate with each other over the public channel and determine whether or not they want to be part of a coalition. However, we do not assume any privilege for coalitions, in particular, if the members of a given coalition need to communicate with each other, they only have access to the aforementioned public communication channel. In the following, \emph{cooperation among a set of agents $\mathcal{S} \subseteq \mathcal{L}$ means that the agents in $\mathcal{S}$ agree on participating in a secret key generation scheme as defined in Definition \ref{definition_modelgs}}.

The questions we are interested in are the following. (i) Can selfish agents find a consensus about which coalitions to form? %
 (ii) If such consensus exists, how should the value, i.e., the secret-key sum-rate, of each coalition be allocated among its agents? 

We define a game corresponding to this problem as follows. For $\mathcal{S} \subseteq \mathcal{L}$, let $\mathfrak{S}(\mathcal{S})$ be the set of all sequences $(S_N(\mathcal{S}))_{N \in \mathbb{N}}$, where $S_N(\mathcal{S})$ is a $((2^{NR_i})_{i\in \mathcal{S}},N)$ secret-key generation strategy as defined in Definition \ref{definition_modelgs}. The set of strategies that coalition $\mathcal{S} \subseteq \mathcal{L}$ can adopt is $\mathfrak{S}(\mathcal{S})$. Consider a sequence of payoff functions $(\pi_l)_{l \in \mathcal{L}}$, where for $l \in \mathcal{L}$, $\pi_l(a_{\mathcal{L}})$ represents the payoff of agent~$l$, i.e., the rate of its secret key, when the strategies $a_{\mathcal{L}} \in \bigcup_{\mathcal{P} \in \mathfrak{P}} \left( \bigtimes_{\mathcal{S}\in\mathcal{P}} \mathfrak{S}(\mathcal{S})\right)$ are  played by the agents, where $\mathfrak{P}$ denotes the set of all partitions of $\mathcal{L}$. We assume a decentralized setting in the sense that the base station does not  influence the strategies of the agents, i.e., is not a player but a passive entity. 
\begin{rem}
Our study aims at modeling selfish constraints for the agents, and how they can be accounted in a decentralized manner without an external authority entity. However, considering an active base station that could force coalition-building is an interesting avenue for future research.
\end{rem}

We next wish to formulate a coalitional game by associating with each coalition of cooperating agents $\mathcal{S} \subseteq \mathcal{L}$ a certain worth $v(\mathcal{S})$. As detailed in Section \ref{sec:analysis}, such mapping $v$ provides with a tool to study the stability of coalitions formed by the agents, where stability of a coalition means that there is no incentive to merge with another coalition or to split into smaller coalitions. Two potential choices for the worth $v(\mathcal{S})$ of coalition $\mathcal{S} \subseteq \mathcal{L}$  are the following,\text{~\cite{aumann1960neumann,jentzsch1964some}} 
\begin{align}
\max_{ \substack{a_{\mathcal{S}} \\ \in \mathfrak{S}(\mathcal{S})}} \min_{ \substack {a_{\mathcal{L} \backslash \mathcal{S}} \\ \in \mathfrak{S}(\mathcal{L} \backslash\mathcal{S})}} \sum_{i \in \mathcal{S}}\pi_i(a_{\mathcal{S}},a_{\mathcal{L} \backslash \mathcal{S}}), \label{eq6}\\
 \min_{ \substack {a_{\mathcal{L} \backslash \mathcal{S}} \\ \in \mathfrak{S}(\mathcal{L} \backslash\mathcal{S})}}  \max_{ \substack{a_{\mathcal{S}} \\ \in \mathfrak{S}(\mathcal{S})}} \sum_{i \in \mathcal{S}}\pi_i(a_{\mathcal{S}},a_{\mathcal{L} \backslash \mathcal{S}}),\label{eq7}
\end{align}
where the quantity in \eqref{eq6} corresponds to the payoff that coalition $\mathcal{S}$ can ensure to its members regardless of the strategies adopted by the member of $\mathcal{L} \backslash \mathcal{S}$, and the one in~\eqref{eq7} to the payoff that coalition $\mathcal{L} \backslash \mathcal{S}$ cannot prevent coalition $\mathcal{S}$ to receive. See, for instance, \cite{shapley1973gameb} for a detailed explanation of the subtle difference between these two notions in general. Observe also that for our problem both quantities are equal since for any $\mathcal{S} \subseteq \mathcal{L}$, there exists $a^*_{\mathcal{L} \backslash \mathcal{S}} \in \mathfrak{S}(\mathcal{L} \backslash\mathcal{S})$ such that for any strategies $a_{\mathcal{S}}\in \mathfrak{S}(\mathcal{S})$, we have  
\begin{align}
 \sum_{i \in \mathcal{S}}\pi_i(a_{\mathcal{S}},a_{\mathcal{L} \backslash \mathcal{S}}) \geq \sum_{i \in \mathcal{S}}\pi_i(a_{\mathcal{S}},a^*_{\mathcal{L} \backslash \mathcal{S}}).
 \end{align}
Indeed, consider $a^*_{\mathcal{L} \backslash \mathcal{S}}$ as the strategies consisting in publicly disclosing $X^N_i$ for all agents $i\in \mathcal{L} \backslash \mathcal{S}$.

To summarize, for a DMS $\left(\mathcal{X}_{\mathcal{L}} \times \mathcal{X}_0 , p_{X_{\mathcal{L}}X_0} \right)$, the secret-key generation problem described in Definitions \ref{definition_modelg}, \ref{def}, when the agents are selfish is cast  as a coalitional games $(\mathcal{L},v)$ where the value function is defined  as
\begin{align} \label{eqvf}
v  : 2^{\mathcal{L}}& \to \mathbb{R}^+, 
 \mathcal{S} \mapsto \max_{ \substack{a_{\mathcal{S}} \\ \in \mathfrak{S}(\mathcal{S})}} \min_{ \substack {a_{\mathcal{L} \backslash \mathcal{S}} \\ \in \mathfrak{S}(\mathcal{L} \backslash\mathcal{S})}} \sum_{i \in \mathcal{S}}\pi_i(a_{\mathcal{S}},a_{\mathcal{L} \backslash \mathcal{S}})
\end{align}
such that for any $\mathcal{S} \subseteq \mathcal{L}$, $v(\mathcal{S})$ corresponds to the maximal secret-key sum-rate achievable by coalition $\mathcal{S}$ when \emph{no specific strategy is assumed} for the agents in $\mathcal{L} \backslash \mathcal{S}$. 

 \section{Game analysis} \label{sec:analysis}
 For any $\mathcal{S} \subseteq \mathcal{L}$, we define the complement of $\mathcal{S}$ as $\mathcal{S}^c \triangleq \mathcal{L} \backslash \mathcal{S}$. 
In Section~\ref{sec:supadd}, we study the properties of the game defined in Section \ref{secmod2} and, in Section \ref{sec:candid}, we propose candidates for the secret-key rate allocation.

\subsection{Properties of the game and characterization of its core} \label{sec:supadd}

We first provide the following characterization of the value function $v$ defined in \eqref{eqvf}. %

\begin{thm} \label{propsk}
We have for any coalition $\mathcal{S} \subseteq \mathcal{L}$
\begin{align}
\max_{ \substack{a_{\mathcal{S}} \\ \in \mathfrak{S}(\mathcal{S})}} \min_{ \substack {a_{\mathcal{S}^c} \\ \in \mathfrak{S}(\mathcal{S}^c)}} \sum_{i \in \mathcal{S}}\pi_i(a_{\mathcal{S}},a_{\mathcal{S}^c}) = I\left(X_{\mathcal{S}}; X_0 | X_{\mathcal{S}^c}  \right).
\end{align}
Hence, for any $\mathcal{S} \subseteq \mathcal{L}$
\begin{align}
v(\mathcal{S}) = I\left(X_{\mathcal{S}}; X_0 | X_{\mathcal{S}^c}  \right).
\end{align}
\end{thm}

\begin{proof}
Consider the secret-key generation problem described in Definitions \ref{definition_modelgs} and \ref{defs}.  $v(\mathcal{S})$ corresponds to the secret-key sum-rate capacity $C_{\mathcal{S}}$ for  coalition $\mathcal{S} \subseteq \mathcal{L}$, i.e., the maximal secret-key sum-rate $\sum_{i\in \mathcal{S}}R_i$ achievable by coalition $\mathcal{S}$. Moreover, we have 
\begin{align} \label{eqcqp}
C_{\mathcal{S}} = I\left(X_{\mathcal{S}}; X_0 | X_{ \mathcal{S}^c} \right).
\end{align}
The converse proof for \eqref{eqcqp} follows from \cite{Maurer93,Ahlswede93} by considering two legitimate users, each observing $X_{\mathcal{S}}^N$ and $X_0^N$, one Eavesdropper, observing $X^N_{ \mathcal{S}^c} $, and by the Markov chain \eqref{Markov}. The achievability part is more involved and will later follow from Corollary \ref{cor} derived in Section \ref{sec:CS}. We intentionally postpone its proof to streamline presentation. 
\end{proof}
We now review the notion of superadditivity.
\begin{defn}
A game $(\mathcal{L},v)$ is superadditive if $v: 2^{\mathcal{L}} \to \mathbb{R}^+$ is such that 
\begin{align}
\forall \mathcal{S},\mathcal{T}\subseteq \mathcal{L}, \mathcal{S} \cap\mathcal{T} = \emptyset \implies  v(\mathcal{S}) + v(\mathcal{T}) \leq v(\mathcal{S} \cup \mathcal{T}) .
\end{align}
\end{defn}

\begin{pro} \label{propod}
The game $(\mathcal{L},v)$ defined in \eqref{eqvf} is superadditive.
\end{pro}
\begin{proof}
Let $\mathcal{S},\mathcal{T}\subseteq \mathcal{L}$, $\mathcal{S} \cap\mathcal{T} = \emptyset$. We have
\begin{subequations}
\begin{align}
& v(\mathcal{S} \cup \mathcal{T}) \nonumber \\
& = I\left(X_{\mathcal{S}\cup \mathcal{T}}; X_0 | X_{\mathcal{S}^c \cap \mathcal{T}^c}  \right)\\
& = I\left(X_{\mathcal{S}}; X_0 | X_{\mathcal{S}^c \cap \mathcal{T}^c}  \right)   + I\left(X_{ \mathcal{T}}; X_0 | X_{ \mathcal{T}^c}  \right)\\
& =  H \left(X_{\mathcal{S}}| X_{\mathcal{S}^c \cap \mathcal{T}^c}  \right) - H \left(X_{\mathcal{S}}| X_0 X_{\mathcal{S}^c }  \right) \nonumber \\
& \phantom{--}+ I\left(X_{ \mathcal{T}}; X_0 | X_{ \mathcal{T}^c}  \right)\label{eqcondmarkov1}\\
&\geq  I\left(X_{ \mathcal{	S}}; X_0 | X_{ \mathcal{S}^c}  \right) + I\left(X_{ \mathcal{T}}; X_0 | X_{ \mathcal{T}^c}  \right) \label{eqcondmarkov}\\
& = v(\mathcal{S}) + v(\mathcal{T}),
\end{align}
where \eqref{eqcondmarkov1} holds by \eqref{Markov}, \eqref{eqcondmarkov} holds because conditioning reduces entropy.
\end{subequations}
\end{proof}
Superadditivity implies that there is an interest in forming a large coalition to obtain a larger secret-key sum rate, however, large coalition might not be in the individual interest of the agents, in the sense that increasing the secret-key sum-rate of a given coalition might not lead to an increased individual secret-key rate for every player in the coalition. A useful concept to overcome this complication is the core of the game.

\begin{defn}[e.g. \cite{maschler1979geometric}] \label{defcore}
The core of a superadditive game $(\mathcal{L},v)$ is defined as follows.
\begin{multline} \label{eqcore}
\mathcal{C}(v) \triangleq \\ \left\{ (R_l)_{l \in \mathcal{L}} : \sum_{l \in \mathcal{L}} R_l = v(\mathcal{L}) \text{ and }\sum_{i \in \mathcal{S}} R_i \geq v(\mathcal{S}), \forall \mathcal{S} \subset \mathcal{L} \right\}.
\end{multline}
\end{defn}

Observe that for any point in the core, the grand coalition, i.e., the coalition $\mathcal{L}$, is in the best interest to all agents, since the set of inequalities in \eqref{eqcore} ensures that no coalition of agents can increase its secret-key sum-rate by leaving the grand coalition. Observe also that for any point in the core the maximal secret-key sum rate $v(\mathcal{L})$ for the grand coalition  is achieved. %
In general, the core of a game can be empty. However, we will show that the game we have defined has a non-empty core. %

Definition \ref{defcore} further clarifies the choice of the value function~$v$. A coalition $\mathcal{S}$ wishes to be associated with a value $v(S)$ as large as possible, while the agents outside $\mathcal{S}$ wish $v(S)$ to be as small as possible to demand a higher share of $v(\mathcal{L})$. The latter achieve their goal by waiving a threat argument, which consists in arguing that they could adopt the strategy that minimizes $v(S)$ by publicly disclosing their source observations, whereas coalition $\mathcal{S}$ achieves its goal by arguing that it can always achieve the secret-key sum-rate capacity of Theorem~\ref{propsk}, irrespective of the strategy of agents in $\mathcal{S}^c$. This formulation is analogous to the one for the Gaussian multiple access channel problem studied in~\cite{la2004game}, and the Gaussian multiple access wiretap channel problem studied in~\cite{ISIT17}, where users can also form coalitions to request a larger communication sum-rate by means of jamming threats, and is generically termed as alpha effectiveness or alpha theory~\cite{aumann1960neumann,jentzsch1964some,shapley1973gameb}.

We now introduce the notion of convexity for a game to better understand the structure of the core of our game.
\begin{defn} [\!\!\cite{shapley1971cores}]
A game $(\mathcal{L},v)$ is convex if $v: 2^{\mathcal{L}} \to \mathbb{R}^+$ is supermodular, i.e., 
\begin{align}
\forall \mathcal{U}, \mathcal{V} \subseteq \mathcal{L}, v(\mathcal{U}) + v(\mathcal{V}) \leq v(\mathcal{U} \cup \mathcal{V}) + v(\mathcal{U} \cap \mathcal{V}).
\end{align}
\end{defn}
The intuition behind this definition is that supermodularity provides a stronger incentive to form coalition than superadditity. Indeed, supermodularity of a function $v: 2^{\mathcal{L}} \to \mathbb{R}^+$ can equivalently be defined as follows \cite{shapley1971cores}
\begin{multline}
\forall l \in \mathcal{L}, \forall \mathcal{T} \subseteq \mathcal{L}\backslash\{l\}, \forall \mathcal{S} \subseteq \mathcal{T},\\  v(\mathcal{S} \cup \{l\}) -  v(\mathcal{S}) \leq v(\mathcal{T} \cup \{l\}) -  v(\mathcal{T}),
\end{multline}
which means that, in addition to superaddivity, the contribution of a single agent to a given coalition increases with the size of the coalition it joins. We also refer to \cite{shapley1971cores} for other interpretations of supermodularity. 

\begin{prop} \label{propconv}
The game $(\mathcal{L},v)$ defined in \eqref{eqvf} is convex.
\end{prop}
\begin{proof}
For any $\mathcal{S} \subseteq \mathcal{L}$, we have 
\begin{subequations}
\begin{align}
I(X_{\mathcal{S}}; X_0 | X_{\mathcal{S}^c})
& = H(X_{\mathcal{S}}| X_{\mathcal{S}^c})- H(X_{\mathcal{S}}| X_{\mathcal{S}^c} X_0)   \\ 
& =  H(X_{\mathcal{S}}| X_{\mathcal{S}^c})- H(X_{\mathcal{S}}|  X_0)\\
& =  H(X_{\mathcal{L}}) - H(X_{\mathcal{S}^c})- H(X_{\mathcal{S}}|  X_0), \label{eqconv}
\end{align}
\end{subequations}
where we have used the Markov chain \eqref{Markov} in the second equality.
Then, $\mathcal{S} \mapsto -H(X_{\mathcal{S}}|  X_0)$ is supermodular because for any $\mathcal{U}, \mathcal{V} \subseteq \mathcal{L}$,
\begin{subequations}
\begin{align}
&H(X_{\mathcal{U}\cup \mathcal{V}}|  X_0) + H(X_{\mathcal{U}\cap \mathcal{V}}|  X_0)\\
& = H(X_{\mathcal{U}}|  X_0) + H(X_{\mathcal{V}\backslash \mathcal{U} }| X_{\mathcal{U}} X_0) + H(X_{\mathcal{U}\cap \mathcal{V}}|  X_0)\\
& \leq H(X_{\mathcal{U}}|  X_0) + H(X_{\mathcal{V}\backslash \mathcal{U} }| X_{\mathcal{U}\cap \mathcal{V}} X_0) \nonumber \\
& \phantom{--} + H(X_{\mathcal{U}\cap \mathcal{V}}|  X_0) \label{eq14}\\
& = H(X_{\mathcal{U}}|  X_0) + H(X_{\mathcal{V}}|  X_0),
\end{align}
\end{subequations}
where \eqref{eq14} holds because conditioning reduces entropy.
Consequently, $\mathcal{S} \mapsto - H(X_{\mathcal{S}^c})$ is also supermodular since for any supermodular function $w$,  $\mathcal{S} \mapsto w(\mathcal{S}^c)$ is supermodular. Hence, by~\eqref{eqconv} we conclude that $v$ is supermodular. 
\end{proof}
A consequence of Proposition \ref{propconv} is that the core of our game is non-empty.

\begin{cor} \label{corcore}
By \cite{shapley1971cores}, any convex game has non-empty core. Hence, by Proposition \ref{propconv}, our game defined in \eqref{eqvf} has a non-empty core $\mathcal{C}(v)$.
\end{cor}

\begin{rem}
From a geometric point of view, $\left\{ (R_l)_{l \in \mathcal{L}} : \sum_{i \in \mathcal{S}} R_i \geq v(\mathcal{S}), \forall \mathcal{S} \subset \mathcal{L} \right\}$ is a contrapolymatroid~\cite{edmonds2003submodular} when $v$ is convex, and its intersection with the hyperplane $\left\{ (R_l)_{l \in \mathcal{L}} : \sum_{l \in \mathcal{L}} R_l = v(\mathcal{L}) \right\}$ forms the core of $v$\cite{shapley1971cores}. See Example \ref{ex} and Figure \ref{fig:core} for an illustration.
\end{rem}
\begin{rem}
In the case of a convex game, the core coincides with the bargaining set for the grand coalition \cite{maschler1971kernel} and thus admits an alternative interpretation, in terms of stable allocations resulting from a sequence of ``threats" and ``counter-threats", see \cite{maschler1971kernel,aumann1964bargaining} for further details.
\end{rem}

We provide an alternative characterization of the core that will turn out to be useful in the following. It can also be viewed as a converse for our problem since the secret-key rate-tuples in the core are upper-bounded.%
\begin{thm} \label{propcore}
The core $\mathcal{C}(v)$ of the game $(\mathcal{L},v)$ defined in \eqref{eqvf} is given by
\begin{align}
& \smash{\left\{ (R_l)_{l \in \mathcal{L}} : \forall \mathcal{S} \subseteq \mathcal{L}, \phantom{\frac{1}{N}} \right.} \nonumber \\
& \left. \phantom{\frac{1}{N}-}  I(X_{\mathcal{S}};X_0) - I(X_{\mathcal{S}};X_{\mathcal{S}^c}) \leq \smash{\sum_{i \in \mathcal{S}}} R_i \leq I(X_{\mathcal{S}};X_0) \right\}.
\end{align}
\end{thm}
\begin{proof}
We have the following equivalences
\begin{subequations}
\begin{align*}
& \left(\sum_{l \in \mathcal{L}} R_l = v(\mathcal{L}) \text{ and }\sum_{i \in \mathcal{S}} R_i \geq v(\mathcal{S}), \forall \mathcal{S} \subset \mathcal{L} \right)\\
&\iff \nonumber \\
& \left( \sum_{i \in \mathcal{S}} R_i = v(\mathcal{L}) -  \sum_{i \in \mathcal{S}^c} R_i \text{ and }\sum_{i \in \mathcal{S}} R_i \geq v(\mathcal{S}), \forall \mathcal{S} \subset \mathcal{L} \right)\\
& \iff \nonumber \\
& \left( v(\mathcal{L}) -  v(\mathcal{S}^c) \geq \sum_{i \in \mathcal{S}} R_i \geq v(\mathcal{S}), \forall \mathcal{S} \subseteq \mathcal{L} \right).
\end{align*}
\end{subequations}
Finally, for any $\mathcal{S} \subseteq \mathcal{L}$, we have
\begin{subequations}
\begin{align}
v(\mathcal{L}) - v(\mathcal{S}^c)
& = I(X_{\mathcal{L}};X_0) - I(X_{\mathcal{S}^c};X_0|X_{\mathcal{S}}) \\
& = I(X_{\mathcal{S}};X_0),
\end{align}
\end{subequations}
and by the Markov chain \eqref{Markov}, \begin{align}
v(\mathcal{S})
 = I(X_{\mathcal{S}};X_0) - I(X_{\mathcal{S}};X_{\mathcal{S}^c}). 
\end{align}
\end{proof}
\subsection{Candidates for the secret-key rate allocation} \label{sec:candid}
Although $\mathcal{C}(v)$ has been shown to be non-empty in Section~\ref{sec:supadd}, a remaining issue is now to choose a specific rate-tuple allocation in the core. Shapley introduced a solution concept to ensure fairness according to the following axioms.

\begin{enumerate}[(i)]
\item Efficiency axiom, i.e., the secret-key sum-rate capacity for the grand coalition is achieved; 

\item Symmetry axiom, i.e., any two agents that equally contribute to any coalition in the sense that for any $i,j \in \mathcal{L}$, for any $\mathcal{S} \subseteq{L}$ such that $i\neq j$ and $i,j\notin \mathcal{S}$, $v(\mathcal{S}\cup \{i\})=v(\mathcal{S}\cup \{j\})$, obtain the same individual secret-key rate;

\item Dummy axiom, i.e., any agent that does not bring value to any coalition he can join, in the sense,  for any $i \in \mathcal{L}$, for any $\mathcal{S} \subseteq{L}$ such that $i \notin \mathcal{S}$, $v(\mathcal{S}\cup \{i\})=v(\mathcal{S})$, receives a null secret-key rate; 

\item Additivity axiom, i.e., for any two games $v$ and $u$ played by the agents, the individual secret-key length obtained by an agent for the game $u+v$, is the sum of secret-key lengths when $u$ and $v$ are played separately. 
\end{enumerate}

We give an interpretation of the additivity axiom in the situation described next.
		 For $i \in \{ 1,2\}$, assume that the agents collect $N$ source observations, denoted by $(X^N_{\mathcal{L},i},X_{0,i}^N)$, for a source with fixed probability distribution $p_{X_{\mathcal{L},i}X_{0,i}}$. For instance, when the source is obtained from channel gains measurements~\cite{wilson2007channel,wallace2010automatic,ye2010information,pierrot2013experimental}, the source statistics will change if the agents randomly change their physical location over time. We also assume that $p_{X_{\mathcal{L},1}X_{0,1}X_{\mathcal{L},2}X_{0,2}}= p_{X_{\mathcal{L},1}X_{0,1}} p_{X_{\mathcal{L},2}X_{0,2}}$. %
Let $i \in \{ 1,2\}$ and define $\bar{i} = 3-i$. We define the game $v_i$ as extracting keys from $(X^N_{\mathcal{L},i},X_{0,i}^N)$ \emph{when the distribution $p_{X_{\mathcal{L},\bar{i}},X_{0,\bar{i}}}$ is unknown}, and we define the game $w$ as extracting keys from $(X^N_{\mathcal{L},1},X_{0,1}^N,X^N_{\mathcal{L},2},X_{0,2}^N)$ \emph{when both distributions  $p_{X_{\mathcal{L},1}X_{0,1}}$ and  $p_{X_{\mathcal{L},2}X_{0,2}}$ are known}. By independence and by Theorem~\ref{propsk}, the value function associated with $w$ is the sum of the value functions of $v_1$ and $v_2$. In this setup, we interpret the additivity axiom as follows.  If the agents extract keys from $(X^N_{\mathcal{L},i},X_{0,i}^N)$ being ignorant of the distribution  $p_{X_{\mathcal{L},\bar{i}},X_{0,\bar{i}}}$, then they obtain the same payoff as if they had to extract keys from $(X^N_{\mathcal{L},1},X_{0,1}^N,X^N_{\mathcal{L},2},X_{0,2}^N)$ with the knowledge of both  $p_{X_{\mathcal{L},1},X_{0,1}}$ and  $p_{X_{\mathcal{L},2},X_{0,2}}$.  	
	Hence, under the same assumptions, the additivity axiom means that even if the agents do not know in advance the number $M$ of  source observation batches from independent sources they are going to obtain, they can,  after obtaining each batch of observations, successively generate keys without this knowledge and obtain the same individual key lengths as if they had waited to obtain the $M$ batches to generate keys.

\begin{ex}
We have the following intuitive property.	If there exist $i,j\in \mathcal{L}$ such that $i\neq j$ and $p_{X_i|X_0}= p_{X_j|X_0}$, then Agent $i$ and Agent~$j$ satisfy the symmetry axiom described above. 
\end{ex}
\begin{proof}
Let $i,j \in \mathcal{L}$, and $\mathcal{S} \subseteq{L}$ such that $i\neq j$ and $i,j\notin \mathcal{S}$. Define $\overline{\mathcal{S}^c} \triangleq \mathcal{S}^c \backslash \{i,j\}$. We have
\begin{subequations}	
\begin{align}
& v(\mathcal{S} \cup \{ i \}) \nonumber \\
& = H(X_{\mathcal{L}}) - H(X_{(\mathcal{S}\cup \{ i \})^c})- H(X_{\mathcal{S} \cup \{ i \}}|  X_0) \label{eq19a}\\
& = H(X_{\mathcal{L}}) - H(X_jX_{\overline{\mathcal{S}^c}})- H(X_i|X_0) - H(X_{\mathcal{S} }|  X_0) \label{eq19b} \\
& = H(X_{\mathcal{L}}) - H(X_iX_{\overline{\mathcal{S}^c}})- H(X_j|X_0) - H(X_{\mathcal{S} }|  X_0) \label{eq19c}\\
& = v(\mathcal{S} \cup \{ j \}), \label{eq19d}
\end{align}
\end{subequations}
where \eqref{eq19a} holds by \eqref{eqconv}, \eqref{eq19b} holds by definition of $\overline{\mathcal{S}^c}$ and by the Markov chain \eqref{Markov}, \eqref{eq19c} holds because by the Markov chain \eqref{Markov} $p_{X_iX_{\overline{\mathcal{S}^c}}X_0} = p_{X_i|X_0} p_{X_{\overline{\mathcal{S}^c}}X_0}= p_{X_j|X_0} p_{X_{\overline{\mathcal{S}^c}}X_0} =  p_{X_jX_{\overline{\mathcal{S}^c}}X_0}$, which implies by marginalization over $X_0$, $p_{X_iX_{\overline{\mathcal{S}^c}}} = p_{X_jX_{\overline{\mathcal{S}^c}}}$, which in turn implies $H(X_iX_{\overline{\mathcal{S}^c}}) = H(X_jX_{\overline{\mathcal{S}^c}})$, \eqref{eq19d} holds similar to \eqref{eq19a} and~\eqref{eq19b}.
\end{proof}

\begin{prop} [e.g. \cite{ichiishi2014game}]
Given a coalitional  game $(\mathcal{L},v)$, there exists a unique $L$-tuple $\left(R^{\textup{Shap}}_l\right)_{l\in \mathcal{L}}$ that satisfies the efficiency, symmetry, dummy, and additivity axiom described above.  $\left(R^{\textup{Shap}}_l\right)_{l\in \mathcal{L}}$ is called the Shapley value.
 \end{prop}
For convex games, the  Shapley value is in the core, and is explicited in the following proposition.

\begin{prop} \label{propshap}
The Shapley value of $(\mathcal{L},v)$ defined in \eqref{eqvf} is in $\mathcal{C}(v)$ and is given by $\forall l \in \mathcal{L}$,
\begin{subequations}
\begin{align}
R_l^{\textup{Shap}} 
& = \sum_{\mathcal{S} \subseteq \mathcal{L} \backslash \{l\}} \frac{|\mathcal{S}|!(L-|\mathcal{S}| -1)!}{L!}  \left( v(\mathcal{S} \cup \{ l \}) - v(\mathcal{S}) \right)\label{eqShap1}\\
& = I\left(X_l;X_0\right) - \frac{1}{L}\sum_{\mathcal{S} \subseteq \mathcal{L} \backslash \{l\}} \dbinom{L-1}{|\mathcal{S}|}^{-1}  I\left(X_l;X_{\mathcal{S}}\right) .\label{eqShap}
\end{align}
\end{subequations}
\end{prop}

\begin{proof}
The fact that the Shapley value belongs to the core follows by~\cite{shapley1971cores} from the convexity of $(\mathcal{L},v)$ proved in Proposition \ref{propconv}. \eqref{eqShap1} is also from \cite{shapley1971cores}. \eqref{eqShap} is obtained by remarking that for any $l \in \mathcal{L}$, for any $\mathcal{S} \subseteq \mathcal{L} \backslash \{l\}$
\begin{subequations}
\begin{align}
&v(\mathcal{S} \cup \{ l \}) - v(\mathcal{S}) \nonumber \\
& = H(X_{\mathcal{S}^c})+ H(X_{\mathcal{S}}|  X_0) - H(X_{(\mathcal{S}\cup \{ l \})^c}) \nonumber \\
& \phantom{--}- H(X_{\mathcal{S} \cup \{ l \}}|  X_0) \label{eq22a}\\
& = H(X_{\mathcal{S}^c\cap \{l\}} | X_{\mathcal{S}^c\backslash \{ l \}})- H(X_l |X_0 X_{\mathcal{S}}) \\
& = H\left(X_l|X_{\mathcal{S}^c \backslash \{l\}}\right) -  H \left(X_l|X_0\right) \label{eq22b}\\
& = I\left(X_l;X_0\right) - I\left(X_l;X_{\mathcal{S}^c \backslash \{l\}}\right),
\end{align}
\end{subequations}
where \eqref{eq22a} holds by \eqref{eqconv}, \eqref{eq22b} holds because $l \notin \mathcal{S}$ and by the Markov chain \eqref{Markov}.
Finally, we conclude by observing that
\begin{multline}
\sum_{\mathcal{S} \subseteq \mathcal{L} \backslash \{l\}} \frac{|\mathcal{S}|!(L-|\mathcal{S}| -1)!}{L!} \\ = \sum_{k=0}^{L-1} \dbinom{L-1}{k} \frac{k!(L-k -1)!}{L!} = 1,
\end{multline}
and that a change of variables yields
\begin{multline}
 \sum_{\mathcal{S} \subseteq \mathcal{L} \backslash \{l\}} \frac{|\mathcal{S}|!(L-|\mathcal{S}| -1)!}{L!}  I\left(X_l;X_{\mathcal{S}^c \backslash \{l\}}\right)\\ = \sum_{\mathcal{S} \subseteq \mathcal{L} \backslash \{l\}} \frac{|\mathcal{S}|!(L-|\mathcal{S}| -1)!}{L!}  I\left(X_l;X_{\mathcal{S}}\right). 
\end{multline}
\end{proof}

\begin{rem}
Geometrically, the Shapley value corresponds to the center of gravity of the vertices of $\mathcal{C}(v)$ \cite{shapley1971cores}.  See Example~\ref{ex} and Figure \ref{fig:core} for an illustration.
\end{rem}

Observe that \eqref{eqShap} quantifies the difference of key length obtained for Agent $l$, $l \in \mathcal{L}$, between the case $L=1$ and the case $L>1$. Note also that the term $\frac{1}{L}\sum_{\mathcal{S} \subseteq \mathcal{L} \backslash \{l\}} \dbinom{L-1}{|\mathcal{S}|}^{-1}  I\left(X_l;X_{\mathcal{S}}\right)$ is upper-bounded by $I(X_{l};X_{\mathcal{L}\backslash \{l\}})$ according to Theorem \ref{propcore} since the Shapley value belongs to the core.

Note that the Shapley value might not always be meaningful as a solution concept. In particular, the additivity axiom might not always be relevant in our problem, for instance, if the agents do not obtain several batches of  observations from sources with independent statistics. 
Finding an axiomatized solution concept that could be universally agreed upon in our setting remains an open problem.

We next discuss the nucleolus as solution concept and one of its non-axiomatized interpretation that has attracted a certain interest in many studies.

\begin{defn} [\!\!\cite{schmeidler1969nucleolus}]
Define the set $\mathcal{Y} \triangleq \{ \mathbf{y} = (y_i)_{i \in \mathcal{L}} \in \mathbb{R}^L_+ : \sum_{i\in \mathcal{L}}y_i = v(\mathcal{L})\}$. For $\mathbf{y} \in \mathcal{Y}$, for $\mathcal{S} \in 2^{\mathcal{L}}$, define the excess $e (\mathbf{y},\mathcal{S}) \triangleq v(\mathcal{S}) - \sum_{i\in \mathcal{S}}y_i$,
and define the vector $\theta (\mathbf{y})= (\theta_i (\mathbf{y}))_{i \in \llbracket 1, 2^L \rrbracket}\in \mathbb{R}^{2^L}$ as $(e (\mathbf{y},\mathcal{S}))_{\mathcal{S} \in 2^{\mathcal{L}}}$ sorted in nonincreasing order, i.e., for $i,j \in \llbracket 1, 2^L \rrbracket, i<j \implies \theta_i (\mathbf{y}) \geq \theta_j (\mathbf{y})$. 
The nucleolus is defined as
\begin{align}
\{ \mathbf{y}_0 \in \mathcal{Y} : \theta (\mathbf{y}_0) \preceq \theta (\mathbf{y}), \forall y \in \mathcal{Y}\},
\end{align}
where $\preceq$ denote the lexicographic order, i.e., for $\mathbf{y}^{(1)}, \mathbf{y}^{(2)} \in \mathcal{Y}$, 
\begin{multline}
\left(\mathbf{y}^{(1)} \preceq \mathbf{y}^{(2)} \right)
\iff \left( \mathbf{y}^{(1)} = \mathbf{y}^{(2)}  \right. \\ \left. \text{ or } \exists i_0 , \left( \forall j < i_0, y^{(1)}_j =y^{(2)}_j  \text{ and } y^{(1)}_{i_0} < y^{(2)}_{i_0} \right) \right).
\end{multline}
\end{defn}
A possible interpretation of the nucleolus is to see the excess $e (\mathbf{y},\mathcal{S}) \triangleq v(\mathcal{S}) - \sum_{i\in \mathcal{S}}y_i$ for some $\mathbf{y} \in \mathcal{Y}$, $\mathcal{S} \in 2^{\mathcal{L}}$, as an indicator of dissatisfaction of coalition $\mathcal{S}$ associated with $\mathbf{y}$ (the higher the excess, the higher the dissatisfaction). One thus might want to choose the $\mathbf{y}$ that minimizes the maximal excess, i.e., the first component of $\theta$. If several choices for $\mathbf{y}$ are possible, one can decide to select $\mathbf{y}$ such that the second largest excess, i.e., the second component of $\theta$, is minimized. One can then continue until a unique choice for $\mathbf{y}$ is obtained as stated in Proposition \ref{propnucl}.
This interpretation appears, for instance, in~\cite{maschler1979geometric}. 
 
 \begin{prop}[\!\!\cite{schmeidler1969nucleolus}] \label{propnucl}
 For a convex game, the nucleolus is a singleton and belongs to the core. 
 \end{prop}

 The nucleolus has, however, no closed-form formula and involves the resolution of successive minimization problems. We illustrate this concept in the following example. For completeness and to compute the nucleolus in Example \ref{ex}, we summarize in Algorithm \ref{alg:nucl} a concise description of the method described in~\cite{behringer1981simplex,fromen1997reducing}.
\begin{rem}
In the case of a convex game, the nucleolus coincides with the kernel~\cite{maschler1979geometric} and thus admits another interpretation, see~\cite[Section 5]{maschler1979geometric} for further details.
\end{rem}

\begin{figure} 
\centering
  \includegraphics[width=8.5cm]{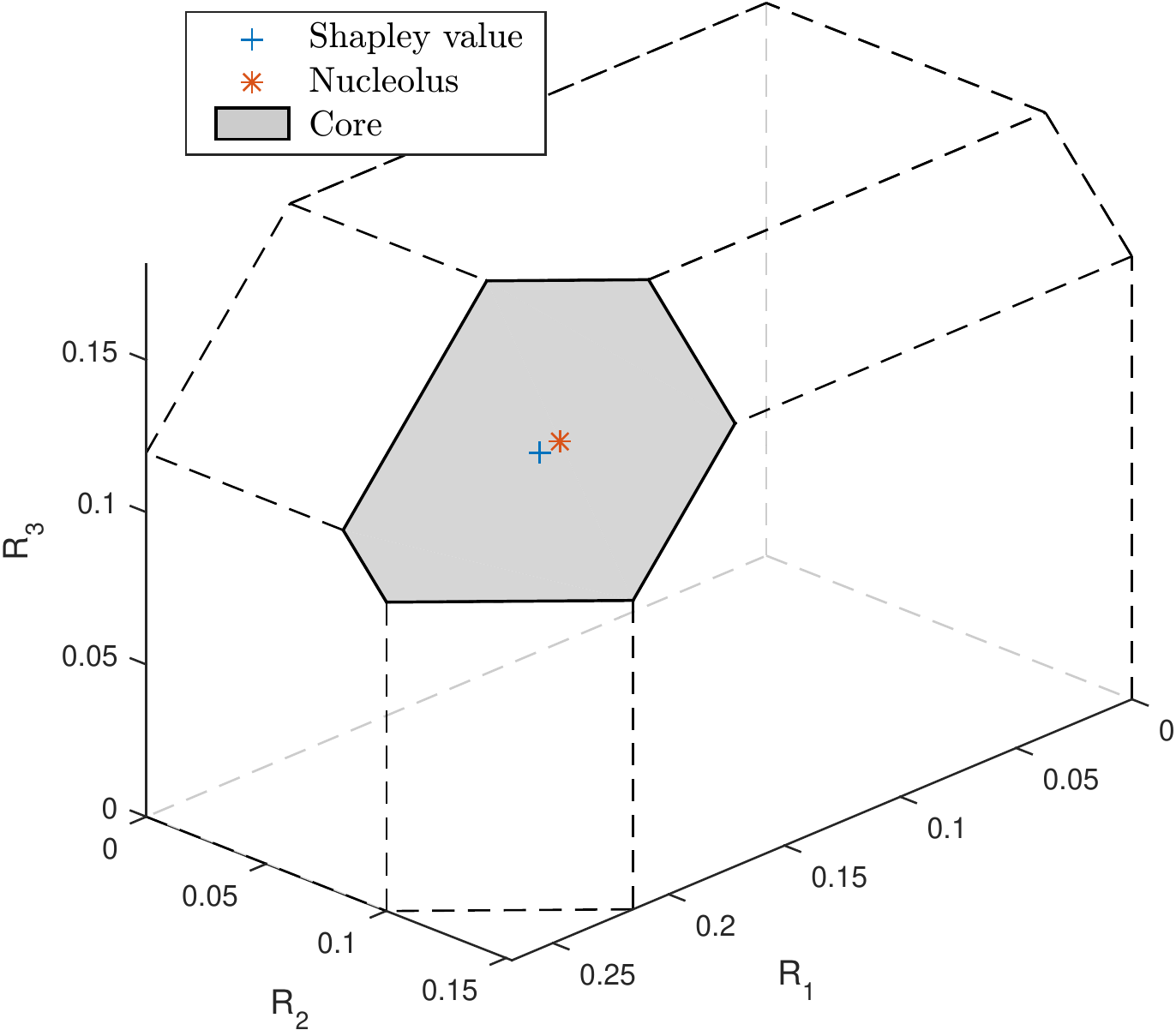}
  \caption{Core, Shapley value, and nucleolus of the game described in Example~\ref{ex}.}  \label{fig:core}
\end{figure}

\begin{algorithm}
  \caption{Nucleolus Computation}
  \label{alg:nucl}
  \begin{algorithmic}[1]
  \STATE $k \leftarrow 0$ 
  \STATE $\mathcal{E}_0 \leftarrow \emptyset$   
\WHILE{the system ($S_k$) has rank $< L$}
\STATE $k \leftarrow k+1$
\STATE Solve the following linear program and let $z_k^*$ denote the value of the objective function obtained
\begin{align*}
 &\text{Minimize } z_k \text{ subject to }  \nonumber \\
 &
\!\!\!\!\!\!\!\!\!\!\!  \left\{ \!\!\!\begin{array}{crr}
        z_k + \displaystyle\sum_{i \in \mathcal{S}}x_i \geq v(\mathcal{S}),  \forall \mathcal{S} \subset {\mathcal{L}} \text{ s.t. } \mathcal{S} \notin \cup_{j=0}^{k-1} \mathcal{E}_j & (E_k)\\
       \left(\!\!\!\! \begin{array}{crr}
        z_j^* + \displaystyle\sum_{i \in \mathcal{S}} x_i = v(\mathcal{S}),  \forall \mathcal{S} \in \mathcal{E}_j, j\in \llbracket 1 ,k-1 \rrbracket  \\
        \displaystyle\sum_{i \in \mathcal{L}}x_i  =v(\mathcal{L})
        \end{array} \!\!\!\right) &  (S_k)
        \end{array}\right.
\end{align*}
\STATE Define $\mathcal{E}_k \triangleq \{ \mathcal{S} \subset \mathcal{L} : (E_k) \text{ holds with equality} \}$
\ENDWHILE 
\RETURN the nucleolus $ \left(R^{\textup{Nucl}}_i \right)_{i \in \mathcal{L}} = \left(x_i \right)_{i \in \mathcal{L}}$
  \end{algorithmic}
\end{algorithm}

\begin{ex} \label{ex}
Let $X_0$ be a Bernoulli random variable with parameter $q\in ]0,1/2[$. Define $X_l \triangleq X_0 \oplus B_l$, $\forall l \in \mathcal{L}$, where the $B_l$'s are independent Bernoulli random variables with parameter $p_l \in ]0,1/2[$. Let $H_b(\cdot)$ denote the binary entropy and define for any $x\in [0,1]$, $\bar{x} = 1-x$.
 For any $\mathcal{S} \subseteq \mathcal{L}$, we have the following formula
 \begin{subequations}
\begin{align}
v(\mathcal{S}) 
& = H(X_{\mathcal{L}}) - H(X_{\mathcal{S}^c})- H(X_{\mathcal{S}}|  X_0) \label{eq25a} \\
& = H(X_{\mathcal{L}}) - H(X_{\mathcal{S}^c})-  \textstyle\sum_{i \in \mathcal{S}} H_b(p_i) \label{eq25b} \\
& =  - \textstyle\sum_{\mathcal{T } \subseteq \mathcal{L}} f_{\mathcal{L}}(\mathcal{T}) \log f_{\mathcal{L}}(\mathcal{T})\nonumber \\
&\phantom{--}  + \textstyle\sum_{\mathcal{T } \subseteq \mathcal{S}^c} f_{\mathcal{S}^c}(\mathcal{T}) \log f_{\mathcal{S}^c}(\mathcal{T}) \nonumber \\
&\phantom{--}-  \textstyle\sum_{i \in \mathcal{S}} H_b(p_i), \label{eq25c}
\end{align}
 \end{subequations}
where \eqref{eq25a} holds by \eqref{eqconv}, \eqref{eq25b} holds by independence of the $B_l$'s, and where in \eqref{eq25c} we have defined for any $\mathcal{S} \subseteq \mathcal{L}$
\begin{align}
f_{\mathcal{S}} : 2^{\mathcal{S}} \to \mathbb{R}_+, \mathcal{T} \mapsto q \prod_{i \in \mathcal{T}} p_i \prod_{j \in \mathcal{S}\backslash \mathcal{T}} \bar{p}_j + \bar{q} \prod_{i \in \mathcal{T}} \bar{p}_i \prod_{j \in \mathcal{S}\backslash \mathcal{T}} p_j.
\end{align}

Assume now that $L=3$, and $(q\text{ }p_1\text{ } p_2\text{ } p_3) = (0.40 \text{ }0.20\text{ } 0.27\text{ } 0.25)$. We obtain $v(\{1\}) \approx 0.17134 $, $v(\{2\}) \approx 0.08205 $, $v(\{3\}) \approx 0.10142 $, $v(\{1,2\}) \approx  0.28771 $, $v(\{1,3\}) \approx 0.31679 $, $v(\{2,3\}) \approx 0.20155 $, $v(\{1,2,3\}) \approx 0.46921 $.
Using Algorithm \ref{alg:nucl} and Proposition \ref{propshap}, we obtain the following secret-key rates
\begin{align*}
R_1^{\textup{Nucl}} \in [0.2109,0.2110],  & \quad R_1^{\textup{Shap}} \in [0.2165,0.2166],\\
R_2^{\textup{Nucl}} \in [0.1172,0.1173], & \quad R_2^{\textup{Shap}} \in [0.1142,0.1143],\\
R_3^{\textup{Nucl}} \in [0.1410,0.1411], & \quad R_3^{\textup{Shap}} \in [0.1384,0.1385].
\end{align*}
The core of the game, as well as the Shapley value and the nucleolus are depicted in Figure~\ref{fig:core}.

\end{ex}

\section{How to achieve any point of the core} \label{sec:CS}
We have seen in Section \ref{sec:analysis} that the grand coalition, i.e., the coalition $\mathcal{L}$, is in the best interest of all agents, and we have characterized the acceptable operating points as the core of the game. Assuming that the grand coalition agrees on an operating point in the core, we now would like to answer whether there exists a secret-key generation protocol for this specific operating point. We show in this section  the following three results. In Theorem \ref{sk}, we claim  that the coding scheme presented in Section~\ref{sec:CS1} achieves for the grand coalition a region that contains the core $\mathcal{C}(v)$. The proof is presented in Section \ref{sec:CS2}. In Theorem~\ref{sk2}, we provide an achievable region for any coalition $\mathcal{S} \subset \mathcal{L}$ of agents. The coding scheme and its analysis partly rely on Theorem \ref{sk} and are discussed in Appendix \ref{App_sk2}. Finally,  we complete the proof of Theorem~\ref{propsk} with Corollary~\ref{cor} obtained from Theorem~\ref{sk2}. 
\begin{thm} \label{sk}
Consider a DMS $\left(\mathcal{X}_{\mathcal{L}} \times \mathcal{X}_0 , p_{X_{\mathcal{L}}X_0} \right)$ such that $\forall l \in \mathcal{L}, |\mathcal{X}_l| = 2$. Any rate tuple in 
\begin{align}
\mathcal{R}_{\mathcal{L}} \triangleq \left\{ (R_l)_{l \in \mathcal{L}} : 0 \leq \sum_{i \in \mathcal{S}} R_i \leq I(X_{\mathcal{S}};X_0), \forall \mathcal{S} \subseteq \mathcal{L} \right\}
\end{align}
is achievable by the grand coalition, in the sense of Definition~\ref{def}, with the coding scheme of Section \ref{sec:CS1}. Moreover, by Theorem \ref{propcore} we have
\begin{align}
\mathcal{R}_{\mathcal{L}} \supseteq \mathcal{C}(v).
\end{align}
\end{thm}

\begin{thm} \label{sk2}
Consider a DMS $\left(\mathcal{X}_{\mathcal{L}} \times \mathcal{X}_0 , p_{X_{\mathcal{L}}X_0} \right)$ such that $\forall l \in \mathcal{L}, |\mathcal{X}_l| = 2$ and the Markov chain \eqref{Markov} holds. Any rate tuple in 
\begin{align}
\mathcal{R}_{\mathcal{S}} \triangleq \left\{ (R_l)_{l \in \mathcal{S}} : 0 \leq \sum_{i \in \mathcal{T}} R_i \leq I(X_{\mathcal{T}};X_0|X_{\mathcal{S}^c}), \forall \mathcal{T} \subseteq \mathcal{S} \right\}
\end{align}
is achievable in the sense of Definition \ref{defs} by the coalition of agents $\mathcal{S} \subset \mathcal{L}$.
\end{thm}
\begin{proof}
	See Appendix \ref{App_sk2}.
\end{proof}

\begin{cor} \label{cor}
Theorem \ref{sk2} implies the achievability part of Theorem \ref{propsk}, i.e., the secret-key sum rate $I(X_{\mathcal{S}}; X_0 | X_{\mathcal{S}^c})$ is achievable by the coalition $\mathcal{S} \subseteq \mathcal{L}$.	
\end{cor}
\begin{proof}
See Appendix \ref{App_cor}.	
\end{proof}

\begin{rem}
Note that Theorem \ref{sk} does not require the Markov chain \eqref{Markov}. Note also that Theorem~\ref{sk} and Theorem \ref{sk2} extend to prime size alphabets by using \cite[Lemma 7]{chou2014polar} in place of Lemma \ref{lemscs2}.
\end{rem}
\begin{rem}
Note that the region in Theorem \ref{sk} has also been shown achievable in \cite{zhangmulti} without the requirements of prime size alphabets. The main difference between Theorem \ref{sk} and~\cite{zhangmulti} is that Theorem \ref{sk} provides an explicit coding scheme. 
\end{rem}

\subsection{Coding Scheme}\label{sec:CS1}
The principle of the coding scheme is to separately deal with reliability and secrecy, as it can be done for secret-key generation between two users \cite{cachin1997linking}, albeit with additional complications. More specifically, a reconciliation step is first performed to allow the base station to reconstruct the observations $X_{\mathcal{L}}^{N}$ of the agents. Then, during a privacy amplification step, each agent extracts from its observations a key that can be reconstructed at the base station. The reconciliation step itself does not present any difficulty, the main complications, compared to a two-user scenario, are (i) to deal with a distributed setting in the privacy amplification step and (ii) to analyze the combination of the reconciliation and privacy amplification steps, as detailed in the next section.

Our coding scheme operates over $B$ blocks of length $N$, where $N$ and $B$ are powers of $2$. We define $\mathcal{B} \triangleq \llbracket 1, B \rrbracket$. We omit indexation of the variables over blocks because encoding is identical for all blocks. The reconciliation step, described in Algoritm \ref{alg:rec}, makes use of polar codes. In particular we introduce the following notation. For $n \in \mathbb{N}$ and $N \triangleq 2^n$, let $G_n \triangleq  \left[ \begin{smallmatrix}
       1 & 0            \\[0.3em]
       1 & 1 
     \end{smallmatrix} \right]^{\otimes n} $ be the source polarization transform defined in~\cite{Arikan10}. For any $l \in \mathcal{L}$, we define  
     \begin{align}
          U_l^N \triangleq X_l^N G_n,
     \end{align}
moreover, for any set $\mathcal{I} \subseteq \llbracket 1,N \rrbracket $, we define $U_l^{N}[\mathcal{I}] \triangleq \left( (U_l)_{i}\right)_{i \in \mathcal{I}}$, where $(U_l)_{i}$ denotes the $i$-th component of the vector $U_l^{N}$, $i \in \mathcal{I}$. For any $l \in \mathcal{L}$, for any $\mathcal{S} \subseteq \mathcal{L}$, we also define the following ``high entropy" and ``very high entropy" sets.
\begin{align}
& \mathcal{H}_{X_l|X_0X_{1:l-1}X_{\mathcal{S}}} \nonumber \\
 &\triangleq 
\left\{ i\in \llbracket 1 ,N \rrbracket \!:\! H\left((U_l)_i | (U_l)^{i-1} X_0^NX^N_{1:l-1}X^N_{\mathcal{S}}\right) \geq \delta_N\right\}, \\
&\mathcal{V}_{X_l|X_0X_{1:l-1}X_{\mathcal{S}}} \nonumber \\
&\triangleq 
\left\{ i\in \llbracket 1 ,N \rrbracket \!:\! H\left((U_l)_i | (U_l)^{i-1} X_0^NX^N_{1:l-1}X^N_{\mathcal{S}}\right) \geq 1- \delta_N\right\}
\end{align}
where we have defined $X_{1:l-1} \triangleq (X_j)_{j \in \llbracket 1:l-1 \rrbracket}$ and $X_{\mathcal{S}} \triangleq (X_j)_{j \in \mathcal{S}}$.
An interpretation of these sets that will be used in our analysis can be summarized in the following two lemmas. %

\begin{lem}[Source coding with side information\cite{Arikan10}] \label{lemscs}
Consider a discrete memoryless source with joint probability distribution $p_{XY}$ over $\mathcal{X}\times\mathcal{Y}$ with $|\mathcal{X}|=2$ and $\mathcal{Y}$ finite. Define $A^{N} \triangleq X^{N} G_n$, and for $\delta_N \triangleq 2^{-N^{\beta}}$ with $\beta \in ]0,1/2[$, the  set $          \mathcal{H}_{X|Y}    \triangleq   \left\{ i \in \llbracket 1, N \rrbracket: H( A_i | A^{i-1} Y^{N})   >   \delta_N  \right\}.
$
Given $A^{N}[\mathcal{H}_{X|Y} ]$ and $Y^{N}$ it is possible to form $\widehat{A}^{N}$ by the successive cancellation (SC) decoder of~\cite{Arikan10} such that $
\lim_{N \to \infty} \mathbb{P} [ \widehat{A}^{N} \neq {A}^{N}] =0.
$ Moreover, $\lim_{N \to \infty} |\mathcal{H}_{X|Y}|/N = H(X|Y)$. 
\end{lem}

\begin{lem}[Privacy amplification\cite{chou2015polar}] \label{lemscs2}
Consider a discrete memoryless source with joint probability distribution $p_{XZ}$ over $\mathcal{X}\times\mathcal{Z}$ with $|\mathcal{X}|=2$ and $\mathcal{Z}$ finite. Define $A^{N} \triangleq X^{N} G_n$, and for $\delta_N \triangleq 2^{-N^{\beta}}$ with $\beta \in ]0,1/2[$, the  set $          \mathcal{V}_{X|Z}    \triangleq   \left\{ i \in \llbracket 1, N \rrbracket: H( A_i | A^{i-1} Z^{N})   >   1-\delta_N  \right\}.
$
$A^{N}[\mathcal{V}_{X|Z} ]$ is almost uniform and independent from $Z^{N}$ in the sense $
\lim_{N \to \infty} \mathbb{V} ( p_{A^{N}[\mathcal{V}_{X|Z} ]Z^{N}} , p_{U} p_{Z^{N}}) =0,
$ where $p_{U}$ is the uniform distribution over $\{0,1\}^{|\mathcal{V}_{X|Z}|}$. Moreover, $\lim_{N \to \infty} |\mathcal{V}_{X|Z}|/N = H(X|Z)$.  
\end{lem}

Hence, by Lemma \ref{lemscs}, the vector $U_l^{N}[\mathcal{H}_{X_l|X_0X_{1:l-1}X_{\mathcal{S}}}]$, $l \in \mathcal{L}$, ensures near lossless reconstruction of $X_l^N$ given $(X_0^N,X^N_{1:l-1},X^N_{\mathcal{S}})$. By Lemma \ref{lemscs2}, the vector $U_l^{N}[\mathcal{V}_{X_l|X_0X_{1:l-1}X_{\mathcal{S}}}]$, $l \in \mathcal{L}$, is almost uniform and independent from $(X_0^N,X^N_{1:l-1},X^N_{\mathcal{S}})$. Note also that by definition $\mathcal{V}_{X_l|X_0X_{1:l-1}X_{\mathcal{S}}} \subset \mathcal{H}_{X_l|X_0X_{1:l-1}X_{\mathcal{S}}}$. We refer to~\cite{chou2015polar,chou2014polar} for further discussion of theses sets.
\begin{algorithm}
  \caption{Reconciliation protocol}
  \label{alg:rec}
  \begin{algorithmic}[1]    
          \FOR{Block $b \in \mathcal{B}$}
 \STATE Define $\breve{X}_{\emptyset}^{N} \triangleq \emptyset$

    \FOR{$l \in \mathcal{L}$}
        \STATE Agent $l$ computes $U_l^{N} \triangleq X_l^{N} G_n$ and transmits $A_l \triangleq U_l^{N}[\mathcal{H}_{X_l|X_0X_{1:l-1}}]$ to the base station over the public channel
\STATE Given $A_{l}$, $X_0^{N}$,  and  $\breve{X}_{1:l-1}^{N}$, the base station forms $\breve{X}_{l}^{N}$ an estimate of $X_{l}^{N}$ using the SC decoder of~\cite{Arikan10}

	\ENDFOR
	
	\STATE Define $A_{\mathcal{L}} \triangleq \left( A_l \right)_{l\in \mathcal{L}}$

\STATE Define $\breve{X}_{\mathcal{L}}^{N} \triangleq (\breve{X}_{l}^{N})_{l\in\mathcal{L}}$

\ENDFOR
	\STATE Let $A^B_{\mathcal{L}}$ denote the total public communication over the $B$ blocks

\STATE The base station and the agents perform a final round of reconciliation on $(\breve{X}_{\mathcal{L}}^{N})^B$ (obtained above) and $({X}_{\mathcal{L}}^{N})^B$ as follows
\STATE Define $\widehat{X}_{\emptyset}^{NB} \triangleq \emptyset$
 \FOR{$l \in \mathcal{L}$}
 \STATE Agent $l \in \mathcal{L}$ computes $V^{NB} \triangleq X_l^{NB}G_{\log_2(NB)}$ and transmits $A_{0,l} \triangleq V^{NB}[\mathcal{H}_{{X}_{l}^{N}|X^N_{1:l-1}\breve{X}_{\mathcal{L}}^{N}}]$ to the base station over the public channel
 \STATE Given $A_{0,l}$, $(\breve{X}_{\mathcal{L}}^{N})^B$, and $\widehat{X}_{1:l-1}^{NB}$, the base station forms $\widehat{X}_{l}^{NB}$ an estimate of $X_{l}^{NB}$ using the SC decoder of~\cite{Arikan10}
\ENDFOR 
 \STATE Define $\widehat{X}_{\mathcal{L}}^{NB} \triangleq (\widehat{X}_{l}^{NB})_{l\in \mathcal{L}}$
\STATE   Define $A_0 \triangleq (A_{0,l})_{l\in\mathcal{L}} $
  \end{algorithmic}
\end{algorithm}

\begin{algorithm}
  \caption{Privacy amplification protocol}
  \label{alg:p}
  \begin{algorithmic}[1] 
              \FOR{Block $b \in \mathcal{B}$}   
    \FOR{Agent $l \in \mathcal{L}$}
        \STATE Compute $K_l \triangleq F_l(X_l^{N})$
        \STATE Publicly transmit the choice of $F_l$ to the base station
	\ENDFOR
	\FOR{$l \in \mathcal{L}$}
	\STATE The base station computes $K_{l}\triangleq F_{l}(X_{l}^{N})$
	\ENDFOR
	\ENDFOR
  \end{algorithmic}
\end{algorithm}

 The privacy amplification step, described in Algorithm \ref{alg:p}, relies on two-universal hash functions~\cite{Carter79,Bennett95}. 
 
 \begin{defn}
 A family $\mathcal{F}$ of two-universal hash functions $\mathcal{F} = \{f:\{0,1 \}^N \to \{0,1\}^r\}$ is such~that
 \begin{align}
 \forall x,x' \in \{0,1 \}^N, x \neq x' \implies \mathbb{P} [F(x)=F(x')] \leq 2^{-r},
 \end{align}
 where $F$ is a function uniformly chosen in $\mathcal{F}$.
 \end{defn}

 For $l \in\mathcal{L}$, we let $F_l : \{0,1\}^N \rightarrow \{0,1\}^{r_l}$, be  uniformly chosen in a family $\mathcal{F}_l$ of two-universal hash functions. Note that $r_l$ represents the key length obtained by Agent~$l$. The main difficulty in the analysis of the privacy amplification step is to find the admissible values, in the sense of Definition \ref{def}, for $r_l$. We leave these quantities unspecified in this section, and will specify them in Section \ref{sec:CS2}.

\subsection{Coding Scheme Analysis}\label{sec:CS2}

\subsubsection{Reconciliation Analysis} \label{sec:recanal}

Line 15 in Algorithm \ref{alg:rec} ensures that for a fixed $N$,
\begin{subequations}
\begin{align}
&\mathbb{P}\left[\widehat{X}_{\mathcal{L}}^{NB} \neq ({X}_{\mathcal{L}}^{N})^B\right] \nonumber \\
 & \leq \sum_{l \in \mathcal{L}} \mathbb{P}\left[\widehat{X}_{l}^{NB} \neq {X}_{l}^{NB} | \widehat{X}_{1:l-1}^{NB} = {X}_{1:l-1}^{NB} \right]  \\
& \xrightarrow{B \to \infty} 0,
\end{align}
\end{subequations}
where the limit follows from Lemma \ref{lemscs}.

\subsubsection{Privacy Amplification Analysis}
We will use the following notation. The indicator function is denoted by $\mathds{1}\{ \omega \}$, which is equal to $1$ if the predicate $\omega$ is true and $0$ otherwise.  For a discrete random variable $X$ distributed according to $p_X$ over the alphabet $\mathcal{X}$, we let 
\begin{align}
&\mathcal{T}^N_{\epsilon}(X)  \triangleq \smash{\left\{ x^N \in \mathcal{X}^N : \left| \frac{1}{N}\smash{\sum_{i=1}^N} \mathds{1}\{x_i = a\} - p_X(a) \right| \nonumber  \right.}\\ 
& \left. \phantom{\frac{1}{N} --------------}\leq \epsilon p_X(a), \forall a \in \mathcal{X} \right\}\end{align}
  denote the $\epsilon$-letter-typical set associated with $p_X$ for sequences of length $N$, see, for instance,~\cite{kramerbook}, and define $\mu_X \triangleq \min_{x \in S_X} p(x)$, where $S_X \triangleq \{ x \in \mathcal{X} : p(x)>0 \}$.  Additionally, the min-entropy of $X$ is defined as
\begin{align}
H_{\infty}(X) \triangleq - \log \left( \max_{x \in \mathcal{X}} p_X(x) \right). 
\end{align}

 We will need the following two lemmas. Lemma~\ref{lem1} is a refined version of \cite{Maurer00} meant to relate a min-entropy to a Shannon entropy, which is easier to study, and Lemma~\ref{lem:lb} can be interpreted as quantifying how much information is revealed about $X_{\mathcal{S}}^N$ knowing the  public communication~$A_{\mathcal{L}}$.
\begin{lem}[{\cite[Lemma 1.1]{chou2014separation}}\cite{Maurer00}] \label{lem1}
Let $\epsilon >0$. Consider a DMS $(\mathcal{X}\times\mathcal{Z},p_{XZ})$ and define the random variable $\Theta$ as

\begin{align}
\Theta & \triangleq \mathds{1} \left\{ (X^{B} , Z^{B}) \in  \mathcal{T}_{2\epsilon}^{B}(XZ) \right\} \mathds{1} \left\{ Z^{B} \in  \mathcal{T}_{\epsilon}^{B}(Z) \right\},
 \end{align}
Then, $\mathbb{P}[\Theta=1]\geq1- \delta_{\epsilon}^0(B)$, with  
\begin{align}\delta_{\epsilon}^0(B) \triangleq 2|S_X|e^{-\epsilon^2 B \mu_X /3}+2|S_{XZ}|e^{-\epsilon^2 B \mu_{XZ} /3}.\end{align} Moreover, if $z^B \in \mathcal{T}_{\epsilon}^{B}(Z)$, then
\begin{multline}
{H}_{\infty}(X^B|Z^B=z^B,\Theta=1) \\  \geq B(1- \epsilon) {H}(X|Z)  + \log( 1-\delta_{\epsilon}^1(B)), 
\end{multline}
where $\delta_{\epsilon}^1(B) \triangleq 2|S_{XZ}|e^{-\epsilon^2 B\mu_{XZ}/6}$.
\end{lem}
\begin{lem} \label{lem:lb}
For any $\mathcal{S} \subseteq \mathcal{L}$, we have
\begin{align}
	 H(X^N_{\mathcal{S}} | A_{\mathcal{L}}) \geq N I(X_{\mathcal{S}};X_0) - o(N).
\end{align}
\end{lem}
\begin{proof}
See Appendix \ref{App_lem2}.	
\end{proof}
We are now equipped to show \eqref{eqmin}.
Let $\mathcal{S} \subseteq \mathcal{L}$ and let 
\begin{align}
\Theta &  \triangleq \mathds{1} \left\{ (X_{\mathcal{L}}^{NB} , A_{\mathcal{L}}^{B}, A_0) \in  \mathcal{T}_{2\epsilon}^{B}(X_{\mathcal{L}}^NA_{\mathcal{L}} A_0) \right\} \nonumber \\
& \phantom{--}\times \mathds{1} \left\{ (A_{\mathcal{L}}^{B},A_0) \in  \mathcal{T}_{\epsilon}^{B}(A_{\mathcal{L}}A_0) \right\}.
\end{align}
Fix $(a_{\mathcal{L}}^{B}, a_0) \in  \mathcal{T}_{\epsilon}^{B}(A_{\mathcal{L}},A_0)$. We define  
\begin{align}
\Omega &  \triangleq \mathds{1} \left\{ {H}_{\infty}(X_{\mathcal{S}}^{NB} | A_{\mathcal{L}}^{B}=a_{\mathcal{L}}^{B}, A_0 = a_0, \Theta=1)\phantom{\sqrt{NB}}\right. \nonumber \\
&\phantom{----} \geq {H}_{\infty}(X_{\mathcal{S}}^{NB} | A_{\mathcal{L}}^{B}=a_{\mathcal{L}}^{B}, \Theta=1)   \nonumber \\
& \left. \phantom{-----}- \textstyle\sum_{l\in\mathcal{L}} |\mathcal{H}_{{X}_{l}^{N}|X^N_{1:l-1}\breve{X}_{\mathcal{L}}^{N}}| -\sqrt{NB} \right\}.
\end{align}
We have 
\begin{subequations}
\begin{align} 
&{H}_{\infty}(X_{\mathcal{S}}^{NB} | A_{\mathcal{L}}^{B}=a_{\mathcal{L}}^{B}, A_0 = a_0,\Theta=1, \Omega =1)   \\ 
& \geq  {H}_{\infty}(X_{\mathcal{S}}^{NB} | A_{\mathcal{L}}^{B}= a_{\mathcal{L}}^{B}, \Theta=1)\nonumber \\
&\phantom{--} - \textstyle\sum_{l\in\mathcal{L}} |\mathcal{H}_{{X}_{l}^{N}|X^N_{1:l-1}\breve{X}_{\mathcal{L}}^{N}}| -\sqrt{NB}\label{eq39000} \\ 
& =  {H}_{\infty}(X_{\mathcal{S}}^{NB} | A_{\mathcal{L}}^{B}= a_{\mathcal{L}}^{B}, \Theta=1) \nonumber \\
&\phantom{--} - \textstyle\sum_{l\in\mathcal{L}} (B {H}({X}_{l}^{N}|X^N_{1:l-1}\breve{X}_{\mathcal{L}}^{N}) + o(B)) -\sqrt{NB}\label{eq3900} \\ 
& =  {H}_{\infty}(X_{\mathcal{S}}^{NB} | A_{\mathcal{L}}^{B}= a_{\mathcal{L}}^{B}, \Theta=1)\nonumber \\
&\phantom{--} - B H({X}_{\mathcal{L}}^{N}|\breve{X}_{\mathcal{L}}^{N}) - o(NB)\\ 
& \geq  {H}_{\infty}(X_{\mathcal{S}}^{NB} | A_{\mathcal{L}}^{B}= a_{\mathcal{L}}^{B}, \Theta=1) \nonumber \\
&\phantom{--} - o(NB)\label{eq390} \\ 
&  \geq B(1- \epsilon) {H}(X_{\mathcal{S}}^{N} | A_{\mathcal{L}})\nonumber \\
&\phantom{--} + \log( 1-\delta_{\epsilon}^1(N,B)) \label{eq39a} - o(NB)\\ 
&  \geq (1- \epsilon) (NB\cdot I(X_{\mathcal{S}};X_0) - o(N)\cdot B) \nonumber \\
&\phantom{--}+ \log( 1-\delta_{\epsilon}^1(N,B)) \label{eq39b} - o(NB)\\
&  = (1- \epsilon) \cdot NB\cdot I(X_{\mathcal{S}};X_0) + \delta^2_{\epsilon}(N,B) , \label{eqmin}
\end{align}
\end{subequations}
where \eqref{eq39000} holds by definition of $\Omega$, \eqref{eq3900} holds by Lemma \ref{lemscs},  \eqref{eq390} holds by Fano's inequality because similar to the analysis in Section \ref{sec:recanal}, $\mathbb{P}[\breve{X}_{\mathcal{L}}^{N} \neq {X}_{\mathcal{L}}^{N}] \xrightarrow{N\to \infty}0$, \eqref{eq39a} holds by Lemma~\ref{lem1} applied to the DMS $(\mathcal{X}_{\mathcal{S}}^N\times \mathcal{A}_{\mathcal{L}}, p_{X^N_{\mathcal{S}} A_{\mathcal{L} }})$, \eqref{eq39b} holds by Lemma~\ref{lem:lb}, and in \eqref{eqmin} we have defined 
\begin{multline} 
\delta^2_{\epsilon}(N,B) \\ \triangleq - o(N)\cdot B (1- \epsilon) + \log( 1-\delta_{\epsilon}^1(N,B))- o(NB).
\end{multline}

\begin{rem}
Unfortunately, unlike the two-user secret-key generation setting considered in~\cite{Maurer00}, it is not possible to use \cite[Lemma 10]{Maurer00} to obtain a tight lower bound on ${H}_{\infty}(X_{\mathcal{S}}^{NB} | A_{\mathcal{L}}^{B}=a_{\mathcal{L}}^{B}) $. 
We use Lemmas \ref{lem1}, \ref{lem:lb} to circumvent this issue.

\end{rem}
We will then need the following version of the leftover hash lemma \cite{haastad1999pseudorandom,Bennett95,dodis2008fuzzy,wullschleger2007oblivious}. 
\begin{lem}  [Leftover hash lemma for concatenated hash functions]\label{lemloh}
Let $X_{\mathcal{L}} \triangleq (X_l)_{l\in \mathcal{L}}$ and $Z$ be random variables distributed according to $p_{X_{\mathcal{L}}Z}$ over $\mathcal{X}_{\mathcal{L}} \times \mathcal{Z}$. For $l \in \mathcal{L}$, let $F_l : \{0,1\}^{n_l} \rightarrow \{0,1\}^{r_l}$, be  uniformly chosen in a family $\mathcal{F}_l$ of two-universal hash functions. Define  $s_{\mathcal{L}} \triangleq\prod_{l\in \mathcal{L}} s_l$, where $s_l \triangleq |\mathcal{F}_l|$, $l \in \mathcal{L}$, and for any $\mathcal{S} \subseteq {\mathcal{L}} $, define $r_{\mathcal{S}} \triangleq\sum_{i\in \mathcal{S}} r_i$. Define also $F_{\mathcal{L}} \triangleq (F_l)_{l \in \mathcal{L}}$ and 
\begin{align} 
F_{\mathcal{L}}(X_{\mathcal{L}}) \triangleq (F_1(X_1)||F_2(X_2)|| \ldots || F_L(X_L)),
\end{align} 
where  $||$ denotes concatenation. Then, for any $z \in \mathcal{Z}$, we have
\begin{align}
\mathbb{V} ( p_{F_{\mathcal{L}}(X_{\mathcal{L}}),F_{\mathcal{L}}|Z=z}, p_{U_{\mathcal{K}}} p_{U_{\mathcal{F}}})  
&\leq   \sqrt{ \sum_{ \substack{ \mathcal{S} \subseteq {\mathcal{L}} \\ \mathcal{S}\neq \emptyset }} 2^{ r_{\mathcal{S}} - H_{\infty}\left( {X_{\mathcal{S}}|Z=z} \right)} 
}, \label{eqh}
\end{align}
where $p_{U_{\mathcal{K}}}$ and $p_{U_{\mathcal{F}}}$ are the uniform distribution over $\llbracket 1, 2^{r_{\mathcal{L}}} \rrbracket$, and $\llbracket 1, s_{\mathcal{L}} \rrbracket$, respectively. 

A consequence of \eqref{eqh} is 
\begin{align}
\mathbb{V} ( p_{F_{\mathcal{L}}(X_{\mathcal{L}}),F_{\mathcal{L}},Z}, p_{U_{\mathcal{K}}} p_{U_{\mathcal{F}}}p_Z)  
&\leq   \sqrt{ \sum_{ \substack{ \mathcal{S} \subseteq {\mathcal{L}} \\ \mathcal{S}\neq \emptyset }} 2^{ r_{\mathcal{S}} - {H}_{\infty}\left( {X_{\mathcal{S}}|Z} \right)} 
}, 
\end{align}
where we have used the average conditional min-entropy of $X$ given $Z$ defined as in \cite{dodis2008fuzzy} by
\begin{align}
  {H}_{\infty}({X|Z}) \triangleq - \log (\mathbb{E}_{Z} \max_{x} p_{X|Z} (x|Z)).
\end{align}
\end{lem}
\begin{proof}
See Appendix \ref{App_lem3}.	
\end{proof}

Combining \eqref{eqmin} and Lemma \ref{lemloh}, we are able to determine the admissible values for $r_l$, $l\in \mathcal{L}$, as follows. 
\begin{subequations}
\begin{align}
&\mathbb{V} ( p_{F_{\mathcal{L}}(X^{NB}_{\mathcal{L}})F_{\mathcal{L}} A^B_{\mathcal{L}} A_0}, p_{U_{\mathcal{K}}} p_{U_{\mathcal{F}}}p_{A^B_{\mathcal{L}} A_0})  \nonumber \\ 
& \leq \mathbb{V} ( p_{F_{\mathcal{L}}(X^{NB}_{\mathcal{L}})F_{\mathcal{L}} A^B_{\mathcal{L}} A_0\Theta \Omega}, p_{U_{\mathcal{K}}} p_{U_{\mathcal{F}}}p_{A^B_{\mathcal{L}} A_0\Theta\Omega})  \label{eq45a} \\ 
& = \mathbb{E} \left[ \mathbb{V} ( p_{F_{\mathcal{L}}(X^{NB}_{\mathcal{L}})F_{\mathcal{L}} A^B_{\mathcal{L}} A_0|\Theta\Omega}, p_{U_{\mathcal{K}}} p_{U_{\mathcal{F}}}p_{A^B_{\mathcal{L}}A_0|\Theta \Omega}) \right] \label{eq45a2}\\ 
& \leq 2\mathbb{P}[\Theta = 0 \lor \Omega=0] +  \nonumber \\
&\phantom{-l} \mathbb{V} ( p_{F_{\mathcal{L}}(X^{NB}_{\mathcal{L}})F_{\mathcal{L}} A^B_{\mathcal{L}} A_0 |\Theta=1, \Omega=1}, p_{U_{\mathcal{K}}} p_{U_{\mathcal{F}}}p_{A^B_{\mathcal{L}} A_0|\Theta=1, \Omega=1})  \label{eq45b} \\ 
& \leq 2\mathbb{P}[\Theta = 0\lor \Omega=0]\nonumber \\
& \phantom{--} +  \mathbb{E} \left[ \mathbb{V} ( p_{F_{\mathcal{L}}(X^{NB}_{\mathcal{L}})F_{\mathcal{L}} |A^B_{\mathcal{L}} A_0 \Theta=1, \Omega=1}, p_{U_{\mathcal{K}}} p_{U_{\mathcal{F}}})  \right] \label{eq45b2}\\ 
&\leq  2\mathbb{P}[\Theta = 0\lor \Omega=0]\nonumber \\
& \phantom{--} +  \mathbb{E}  \sqrt{ \sum_{ \substack{ \mathcal{S} \subseteq {\mathcal{L}} \\ \mathcal{S}\neq \emptyset }} 2^{ r_{\mathcal{S}} - H_{\infty}\left( {X_{\mathcal{S}}^{NB}|A^B_{\mathcal{L}} =a^B_{\mathcal{L}} ,A_0=a_0, \Theta = 1, \Omega=1} \right)} } \label{eq45c} \\ 
&\leq   2\mathbb{P}[\Theta = 0\lor \Omega=0] \nonumber \\
& \phantom{--} +  \mathbb{E}  \sqrt{ \sum_{ \substack{ \mathcal{S} \subseteq {\mathcal{L}} \\ \mathcal{S}\neq \emptyset }} 2^{ r_{\mathcal{S}} - (1- \epsilon) \cdot NB\cdot I(X_{\mathcal{S}};X_0) - \delta^2_{\epsilon}(N,B)}}  \label{eq45d} \\ 
& =   2\mathbb{P}[\Theta = 0\lor \Omega=0] \nonumber \\
& \phantom{--} +   \sqrt{ \sum_{ \substack{ \mathcal{S} \subseteq {\mathcal{L}} \\ \mathcal{S}\neq \emptyset }} 2^{ r_{\mathcal{S}} - (1- \epsilon) \cdot NB\cdot I(X_{\mathcal{S}};X_0) - \delta^2_{\epsilon}(N,B)}} \\
& \leq   2\delta_{\epsilon}^0(N,B) + 2\cdot 2^{-\sqrt{NB}} \nonumber \\
& \phantom{--} +   \sqrt{ \sum_{ \substack{ \mathcal{S} \subseteq {\mathcal{L}} \\ \mathcal{S}\neq \emptyset }} 2^{ r_{\mathcal{S}} - (1- \epsilon) \cdot NB\cdot I(X_{\mathcal{S}};X_0) - \delta^2_{\epsilon}(N,B)}} , \label{eqf} 
\end{align}
\end{subequations}
where \eqref{eq45a} holds by marginalization over $\Theta$ and the triangle inequality, in \eqref{eq45a2} the expectation is with respect to $p_{\Theta,\Omega}$, \eqref{eq45b} holds because $\mathbb{V}(\cdot,\cdot)$ is upper bounded by~$2$, in \eqref{eq45b2} the expectation is with respect to $p_{A^B_{\mathcal{L}}A_0|\Theta=1, \Omega=1}$, \eqref{eq45c} holds by Lemma \ref{lemloh} with the substitutions $z \leftarrow (a^B_{\mathcal{L}}, a_0, \Theta = 1, \Omega =1)$ and $X_{\mathcal{L}} \leftarrow X^{NB}_{\mathcal{L}}$ and the expectation is with respect to $p_{A^B_{\mathcal{L}}A_0|\Theta=1, \Omega=1}$, \eqref{eq45d} holds by \eqref{eqmin} and the expectation is with respect to $p_{A^B_{\mathcal{L}}A_0|\Theta=1, \Omega=1}$, \eqref{eqf} holds by the union bound, Lemma~\ref{lem1}, the definition of $\Omega$ and \cite[Lemma 10]{Maurer00}.

Finally, we conclude that Theorem \ref{sk} holds by Remark \ref{rem} and \eqref{eqf}.

\section{Non-degraded source case when $L=2$} \label{sec:ndg}

In this section, we consider  the setting described in Section~\ref{sec:def} when $L=2$, public communication is restricted to be one-way from the agents to the base station, and when \eqref{Markov} does not hold. Similar to Section~\ref{sec:def}, the value function $v^*$ of this game is defined as  
\begin{align}  \label{eqvfnd}
v^*  : 2^{\mathcal{L}}& \to \mathbb{R}^+, 
 \mathcal{S} \mapsto \max_{ \substack{a_{\mathcal{S}} \\ \in \mathfrak{S}(\mathcal{S})}} \min_{ \substack {a_{\mathcal{L} \backslash \mathcal{S}} \\ \in \mathfrak{S}(\mathcal{L} \backslash\mathcal{S})}} \sum_{i \in \mathcal{S}}\pi_i(a_{\mathcal{S}},a_{\mathcal{L} \backslash \mathcal{S}})
\end{align}
such that for any $\mathcal{S} \subseteq \mathcal{L}$, $v(\mathcal{S})$ corresponds to the maximal secret-key sum-rate achievable by coalition $\mathcal{S}$ when \emph{no specific strategy is assumed} for the agents in $\mathcal{L} \backslash \mathcal{S}$.

Next, we characterize the value function $v^*$ by providing a counterpart to Theorem \ref{propsk}.
\begin{prop} \label{propndg}
We have 
\begin{subequations}
\begin{align}
& v^*(\{1\}) \nonumber \\
& = \max_{V_1-U_1-X_1-(X_0,X_2)} [I(U_1;X_0|V_1) - I(U_1;X_2|V_1)]^+,\\
&v^*(\{2\}) \nonumber \\
& = \max_{V_2-U_2-X_2-(X_0,X_1)} [I(U_2;X_0|V_2) - I(U_2;X_1|V_2)]^+,
\\
&v^*(\{1,2\}) \nonumber \\
& = I(X_1X_2;X_0) ,
\end{align}
\end{subequations}
where $[x]^+ \triangleq \max(x,0)$ for any $x \in \mathbb{R}$.
\end{prop}
\begin{proof}
$v^*(\{1\})$ and $v^*(\{2\})$ are obtained from \cite{Ahlswede93}. $v^*(\{1,2\})$ is obtained from \cite{lai2013simultaneously,zhangmulti}.
\end{proof}
\begin{rem}
If $L>2$, then $v^*(\mathcal{L}) = I(X_{\mathcal{L}};X_0)$ by \cite{zhangmulti}, and for any $l \in \mathcal{L}$, $v^*(\{l\})$ can be obtained from \cite{Ahlswede93}. However, for any $\mathcal{S} \subset \mathcal{L}$ such that $|\mathcal{S}|>1$, a closed form expression for $v^*(\mathcal{S})$ is unknown.
\end{rem}
\begin{rem}
In the case of two-way communication between the base station and the agents, $v^*(\{1,2\})$ has the same expression. However, in this case, no closed-form expression is known for $v^*(\{1\})$ or $v^*(\{2\})$, which both correspond to a secret-key capacity between two parties in presence of an eavesdropper~\cite{Maurer93,Ahlswede93}.
\end{rem}

\begin{pro} \label{proadd2}
	The game $(\mathcal{L},v^*)$ defined in \eqref{eqvfnd} is superadditive.
\end{pro}
\begin{proof}
Any two disjoint coalitions $\mathcal{S},\mathcal{T}\subseteq \mathcal{L}$, $\mathcal{S} \cap\mathcal{T} = \emptyset$, obtain 
secret-key sum-rate capacities that cannot add up to a quantity strictly larger than the secret-key sum-rate capacity of the coalition $\mathcal{S}\cup\mathcal{T}$. Note indeed that the reliability, secrecy, and uniformity constraints for coalitions  $\mathcal{S}$ and $\mathcal{T}$ imply a reliability, secrecy, and uniformity constraint for the coalition $\mathcal{S} \cup \mathcal{T}$. Indeed, we have
\begin{align}
\mathbb{P}[ \widehat{K}_{\mathcal{S}\cup \mathcal{T}} \neq K_{\mathcal{S}\cup \mathcal{T}}] \leq \mathbb{P}[ \widehat{K}_{\mathcal{S}} \neq K_{\mathcal{S}}]  + \mathbb{P}[ \widehat{K}_{\mathcal{T}} \neq K_{\mathcal{T}}].
\end{align}
Next, we have 
\begin{subequations} \label{eqreason}
\begin{align}
&I\left(K_{\mathcal{S}\cup \mathcal{T}} ; A_{\mathcal{S} \cup \mathcal{T}} X^N_{(\mathcal{S} \cup \mathcal{T})^c}\right) \nonumber\\
& = I\left(K_{\mathcal{S}} ; A_{\mathcal{S}} A_{\mathcal{T}} X^N_{(\mathcal{S} \cup \mathcal{T})^c}\right) \nonumber \\
&\phantom{--}+ I\left(K_{\mathcal{T}} ; A_{\mathcal{S}} A_{\mathcal{T}} X^N_{(\mathcal{S} \cup \mathcal{T})^c} |K_{\mathcal{S}}\right) \label{eq91} \\
& \leq I\left(K_{\mathcal{S}} ; A_{\mathcal{S}} A_{\mathcal{T}} X^N_{(\mathcal{S} \cup \mathcal{T})^c} \right) \nonumber \\
&\phantom{--}+ I\left(K_{\mathcal{T}} ;  A_{\mathcal{T}} A_{\mathcal{S}} X^N_{(\mathcal{S} \cup \mathcal{T})^c} K_{\mathcal{S}}\right) \label{eq9a}\\
&\leq  I\left(K_{\mathcal{S} } ; A_{\mathcal{S} } X^N_{\mathcal{S}^c }\right) + I\left(K_{\mathcal{T} } ; A_{\mathcal{T} } X^N_{\mathcal{T}^c}\right),\label{eq9b}
\end{align}
\end{subequations}
where in \eqref{eq91} we decompose $A_{\mathcal{S} \cup \mathcal{T}}$ in $A_{\mathcal{S}}$ and $A_{\mathcal{T}}$,  the public communication emitted by the agents in $\mathcal{S}$ and $\mathcal{T}$, respectively, \eqref{eq9a} holds by positivity of the mutual information and the chain rule, \eqref{eq9b} holds because $(A_{\mathcal{S}},A_{\mathcal{T}}, X^N_{(\mathcal{S} \cup \mathcal{T})^c})$ is a function of $(A_{\mathcal{S}},X^N_{\mathcal{S}^c})$ and $(A_{\mathcal{T}},A_{\mathcal{S}}, X^N_{(\mathcal{S} \cup \mathcal{T})^c} ,K_{\mathcal{S}})$ is a function of $(A_{\mathcal{T}},X^N_{\mathcal{T}^c})$.

Finally, we have
\begin{subequations}
\begin{align}
&\log |\mathcal{K}_{\mathcal{S}\cup \mathcal{T}} | - H(K_{\mathcal{S}\cup\mathcal{T}}) \nonumber \\
&= \log |\mathcal{K}_{\mathcal{S}}| - H(K_{\mathcal{S}})  + \log |\mathcal{K}_{\mathcal{T}}| - H(K_{\mathcal{T}})  \nonumber \\
&\phantom{--}+ I(K_{\mathcal{S}} ; K_{\mathcal{T}}) \\
& \leq \log |\mathcal{K}_{\mathcal{S}}| - H(K_{\mathcal{S}})  + \log |\mathcal{K}_{\mathcal{T}}| - H(K_{\mathcal{T}}) \nonumber \\
&\phantom{--} + I\left(K_{\mathcal{S} } ; X^N_{\mathcal{S}^c }\right).
\end{align}
\end{subequations}

\end{proof}
Note that the proof of Property \ref{proadd2} is valid for any $L$, not only $L=2$. Hence, we can define the core of the game as in Definition \ref{defcore}. From Property \ref{proadd2}, we  immediately obtain the following property since superadditivity implies convexity for $L=2$. 
\begin{pro} \label{proconvex2}
	The game $(\mathcal{L},v^*)$ defined in \eqref{eqvfnd} is convex.
\end{pro}

Next, from \cite{shapley1971cores} and Property \ref{proconvex2}, we deduce the following corollary.
 
\begin{cor}
The game $(\mathcal{L},v^*)$ defined in \eqref{eqvfnd} has a non-empty core.
\end{cor}

As an example of solution concept, we characterize the Shapley value discussed in Section~\ref{sec:candid} for $(\mathcal{L},v^*)$.

\begin{ex}
The Shapley value of the game $(\mathcal{L},v^*)$ defined in~\eqref{eqvfnd} is given for $i\in \{ 1 ,2 \}$ and $\bar{i} \triangleq 3-i$ by
\begin{multline}
R_i^{\textup{Shap}}  = \frac{1}{2}I(X_1X_2;X_0) \\+\frac{1}{2} \displaystyle\max_{V_i-U_i-X_i-(X_0,X_{\bar{i}})} \left[I(U_i;X_0|V_i) - I(U_i;X_{\bar{i}}|V_i)\right]^+ \\  - \frac{1}{2}\max_{V_{\bar{i}}-U_{\bar{i}}-X_{\bar{i}}-(X_0,X_i)} \left[ I(U_{\bar{i}};X_0|V_{\bar{i}}) - I(U_{\bar{i}};X_i|V_{\bar{i}}) \right]^+\label{eqShap1b}.
\end{multline}
\end{ex}

\section{Extension to multiple levels of security clearance} \label{secext}

We now consider that multiple levels of security clearance exist in our model. Specifically, each agent has a pre-defined security clearance level, and it is required that keys generated by agents at a given level must be be kept secret from the agents at a strictly superior level. Note that this setting is related to the problem of simultaneously generating private and secret keys~\text{\cite{zhang2014capacity,ye2005secret}}.

\subsection{Model}
Let $Q\in \mathbb{N}^*$ and define $\mathcal{Q} \triangleq \llbracket 1, Q \rrbracket$. For $q \in \mathcal{Q}$, let $L_q \in \mathbb{N}^*$ and let $\mathcal{L}_q$ be a set of $L_q$ agents. In the following, we consider $Q$ sets $ (\mathcal{L}_q)_{q\in \mathcal{Q}}$ of agents and one base station depicted in Figure~\ref{fig:modellevel}. We also use the notation $\mathcal{L}_{\mathcal{Q}} \triangleq \bigcup_{q \in \mathcal{Q}} \mathcal{L}_q $ to denote all the agents in the $Q$ sets.

\subsubsection{Definition of the source model}Define $\mathcal{X}_{\mathcal{L}_{\mathcal{Q}}}$ as the Cartesian product of $\sum_{q=1}^Q L_q$ finite alphabets $\mathcal{X}_{l}$, $l\in \mathcal{L}_{\mathcal{Q}}$. Consider a discrete memoryless source (DMS) $\left(\mathcal{X}_{\mathcal{L}_{\mathcal{Q}}} \times \mathcal{X}_0 , p_{X_{\mathcal{L}_{\mathcal{Q}}}X_0} \right)$, where $\mathcal{X}_0$ is a finite alphabet and $X_{\mathcal{L}_{\mathcal{Q}}} \triangleq (X_{l})_{l \in \mathcal{L}_{\mathcal{Q}}}$. For $l \in\mathcal{L}_{\mathcal{Q}}$, Agent $l$ observes the component $X_{l}$ of the DMS, and the base station observes the component $X_0$. The source is assumed to follow the following Markov chain: for any $\mathcal{S},\mathcal{T} \subset \mathcal{L}_{\mathcal{Q}}$ such that $\mathcal{S} \cap \mathcal{T} = \emptyset$,
\begin{align} \label{Markovmult}
X_{\mathcal{S}} - X_0 - X_{\mathcal{T}}.
\end{align}

The source's statistics are assumed known to all parties, and communication is allowed over an authenticated  noiseless  public channel. 

\subsubsection{Description of the objectives for the agents} The goal of Agent $l\in \mathcal{L}_{\mathcal{Q}}$ is to generate an individual secret-key with the base station. The index $q \in\mathcal{Q}$ is meant to describe different sets of agents that do not have the same security constraints. In particular, we require that for any $q \in \mathcal{Q}$, the keys generated by the agents in $\mathcal{L}_{q}$ are secret, in an information-theoretic sense, from the agents in $\bigcup_{i \in \llbracket q+1 , Q \rrbracket} \mathcal{L}_i$ but need not to be secret from the agents in $\bigcup_{i \in \llbracket 1 , q \rrbracket} \mathcal{L}_i$. One can interpret it in terms of levels of security clearance, where $\mathcal{L}_q$, $q \in \mathcal{Q}$, represents a set of agents that share the same level of security clearance, and for $q'<q$, $\mathcal{L}_{q'}$, represents another set of agents with a higher security clearance than $\mathcal{L}_{q}$.
\begin{rem}
We require {information-theoretic security} across security clearance levels, meaning that the agents in $\mathcal{L}_q$, $q \in \mathcal{Q}$, must keep their keys secret from all agents in $(\mathcal{L}_i )_{i \in \llbracket q+1 , Q \rrbracket}$, {for any communication strategy} the latter group of agents may decide to adopt. Consequently, 
we discard the possibility for the agents in $(\mathcal{L}_i )_{i \in \llbracket q+1 , Q \rrbracket}$ and $\mathcal{L}_q$ to agree on participating in a common secret key generation scheme. Not doing so might imply that the agents in $(\mathcal{L}_i )_{i \in \llbracket q+1 , Q \rrbracket}$ follow a pre-determined communication strategy, which would contradict the information-theoretic security requirement.
\end{rem}
\begin{figure} 
\centering
\subfloat[Overview of the $Q$ sets of Agents]{\includegraphics[width=8.8cm]{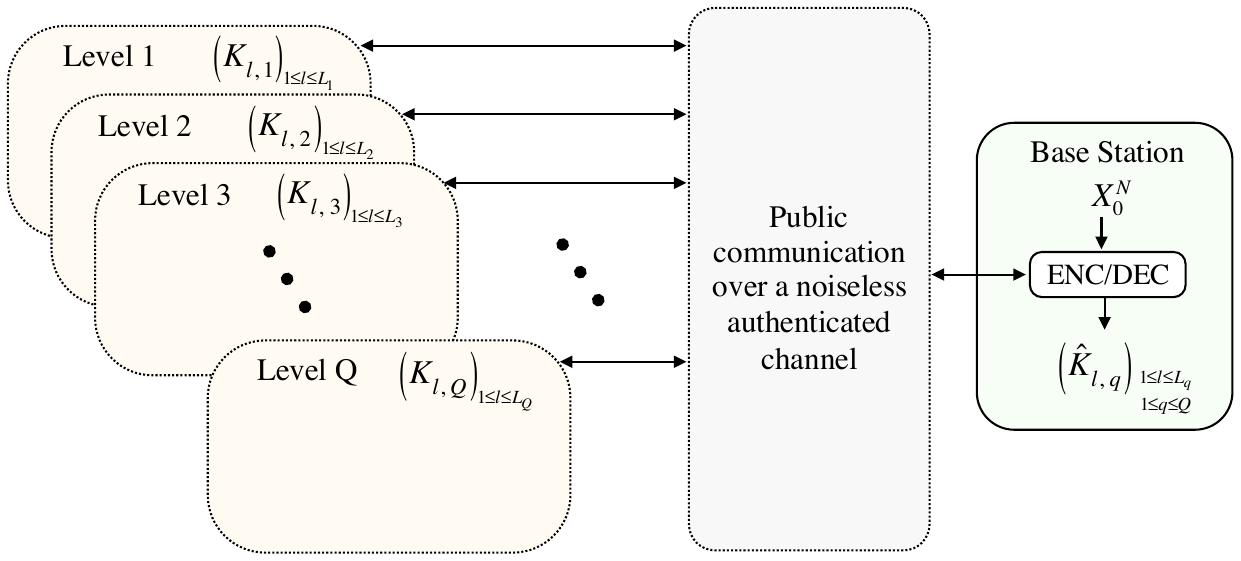}}
	  \label{fig:model}
\subfloat[Representation of the setting at Level $q \in \mathcal{Q}$.]{\includegraphics[width=7.31cm]{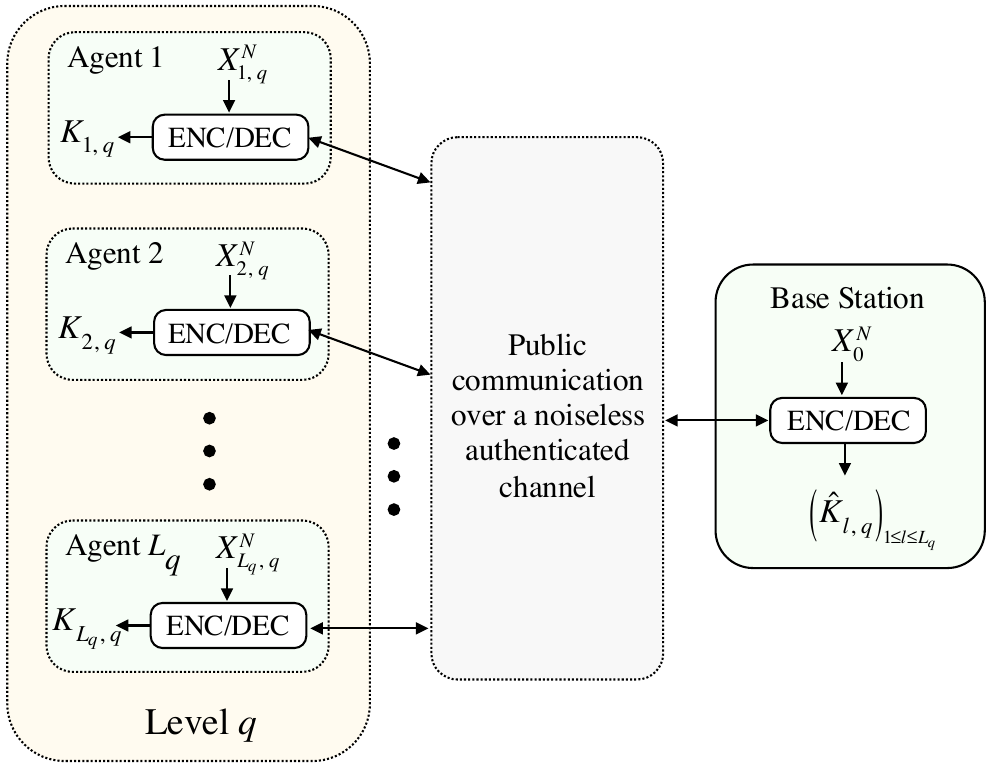}}  \label{fig:model2}
  \caption{Many-to-one secret-key generation setting with multiple levels of security clearance. (b) provides details of the setting described in (a) at a given level $q \in \mathcal{Q}$.}  \label{fig:modellevel}
\end{figure}

To extend Sections \ref{sec:analysis} and \ref{sec:CS} to this setting with multiple clearance levels, we first consider the following auxiliary setting.  We consider the setting described in Definitions~\ref{definition_modelg} and~\ref{def} when an eavesdropper that observes the public communication is also in possession of correlated source observations. The auxiliary setting is presented in Section \ref{sec:mode}. 
Next, to address the multiple levels of security constraints, it is sufficient to  apply the results of Section~\ref{sec:mode} to each security clearance level $q \in \mathcal{Q}$ by considering for the agents in $\mathcal{L}_q$, the DMS $\left(\mathcal{X}_{\mathcal{L}_q} \times \mathcal{X}_0 \times \mathcal{X}_{ \mathcal{L}_{q+1:Q}}\right)$, where $\mathcal{L}_{q+1:Q} \triangleq  \bigcup_{i \in \llbracket q+1,Q \rrbracket} \mathcal{L}_i$, with the assumption 
that an eavesdropper observes the components $X_{\mathcal{L}_{q+1:Q}}$  of the DMS.

\subsection{Auxiliary setting} \label{sec:mode}

Consider a DMS $\left(\mathcal{X}_{\mathcal{L}} \times \mathcal{X}_0 \times \mathcal{Z}, p_{X_{\mathcal{L}}X_0Z} \right)$, where $Z$ is an additional, compared to the source model considered in  Section~\ref{secmod}, component of the source observed by an eavesdropper.  The source is assumed to follow the following Markov chain: for any $\mathcal{S},\mathcal{T} \subseteq \mathcal{L}$ such that $\mathcal{S} \cap \mathcal{T} = \emptyset$,
\begin{align} \label{Markov2}
X_{\mathcal{S}} - X_0 - (X_{\mathcal{T}}Z).
\end{align}
We then consider the same Definitions \ref{definition_modelgs}, \ref{defs} as in Section \ref{secmod3}, where the secrecy constraint~\eqref{eqSec} of Definition \ref{defs} becomes for a coalition $\mathcal{S} \subset \mathcal{L}$
\begin{align}
	\lim_{N \to \infty} I(K_{\mathcal{S}} ; A_{\mathcal{S}} X_{\mathcal{S}^c}^{N} Z^N ) &= 0 .
\end{align} 

 Next, we define our object of study as the secret-key generation problem we have just defined when $\mathcal{S} = \mathcal{L}$ and when the users are selfish. Similar to Section \ref{secmod2}, we wish to understand whether the agents can find a consensus about the coalitions to form, and how the secret-sum rate of each coalition should be allocated among its agents. Following Section \ref{secmod2}, we cast the problem as a coalitional game $(\mathcal{L},v^Z)$ where the value of coalition $\mathcal{S} \subseteq \mathcal{L}$ is defined as  the \emph{maximal 
 secret-key sum-rate} that coalition $\mathcal{S}$ can obtain \emph{regardless of the strategies adopted by the member of $\mathcal{S}^c$}.

Similar to the proof of Theorem \ref{propsk} by using Corollary~\ref{cor4}, stated below, in place of Corollary~\ref{cor}, one can show the following characterization of the value function  $v^Z$.
\begin{align} \label{eqZ}
v^Z  : 2^{\mathcal{L}}& \to \mathbb{R}^+, 
 \mathcal{S} \mapsto I(X_{\mathcal{S}}; X_0 | X_{\mathcal{S}^c}Z).
\end{align}

Similar to Proposition \ref{propconv}, one can show that the game $(\mathcal{L},v^Z)$ is convex, and similar to Theorem~\ref{propcore} that its core is given by
\begin{align}
&\mathcal{C}(v^Z) = \smash{\left\{ (R_l)_{l \in \mathcal{L}} : \forall \mathcal{S} \subseteq \mathcal{L}, \phantom{\frac{1}{N}} \right.} \nonumber \\
& \left. \phantom{\frac{1}{N}----}  I(X_{\mathcal{S}}; X_0 | X_{\mathcal{S}^c}Z) \leq \smash{\sum_{i \in \mathcal{S}}} R_i \leq I(X_{\mathcal{S}};X_0|Z) \right\}.
\end{align}

Moreover, similar to Proposition \ref{propshap}, the Shapley value is in $\mathcal{C}(v^Z)$ and given by
\begin{multline}
\forall l \in \mathcal{L},R_l^{\textup{Shap}} 
\\ = I\left(X_l;X_0|Z\right) - \frac{1}{L}\sum_{\mathcal{S} \subseteq \mathcal{L} \backslash \{l\}} \dbinom{L-1}{|\mathcal{S}|}^{-1}  I\left(X_l;X_{\mathcal{S}}|Z\right) .
\end{multline}

Finally, it is possible to achieve any point of the core $\mathcal{C}(v^Z)$ with an explicit coding scheme, by deducing from Theorem \ref{sk2} the following corollary.
\begin{cor} \label{cor2}
Consider a DMS $\left(\mathcal{X}_{\mathcal{L}} \times \mathcal{X}_0 \times {\mathcal{Z}}, p_{X_{\mathcal{L}}X_0Z} \right)$ such that $\forall l \in \mathcal{L}, |\mathcal{X}_l| = 2$ and the Markov chain \eqref{Markov2} holds. Any rate tuple in 
\begin{align}
&\mathcal{R}^Z_{\mathcal{S}} \triangleq \nonumber \\
&\phantom{-} \left\{ (R_l)_{l \in \mathcal{S}} : 0 \leq \sum_{i \in \mathcal{T}} R_i \leq I(X_{\mathcal{T}};X_0|X_{\mathcal{S}^c}Z), \forall \mathcal{T} \subseteq \mathcal{S} \right\}
\end{align}
is achievable  by the coalition of agents $\mathcal{S} \subseteq \mathcal{L}$. Morevover,
\begin{align}
\mathcal{C}(v^Z) \subset \mathcal{R}^Z_{\mathcal{L}}.
\end{align}
\end{cor}
Finally, from Corollary \ref{cor2}, we deduce, similar to the proof of Corollary \ref{cor} the following corollary, which allows us to establish  \eqref{eqZ}. 
\begin{cor} \label{cor4}
Corollary \ref{cor2} implies that for any $\mathcal{S} \subseteq \mathcal{L}$,  the secret-key sum rate $I(X_{\mathcal{S}}; X_0 | X_{\mathcal{S}^c}Z)$ is achievable by Coalition $\mathcal{S}$.	
\end{cor}

\section{Concluding Remarks} \label{secconcl}
We have studied a pairwise secret-key generation source model between multiple agents and a base station. Although cooperation among agents can increase their individual key length, it can, at the same time, lead to conflict of interests between agents. We have proposed an integrated information-theoretic and game-theoretic formulation of the problem. Specifically, we have cast the problem as a coalitional game in which the value function is determined under  information-theoretic guarantees, i.e., the value associated with a coalition is computed with no restrictions  on the strategies that the users outside the coalition could adopt. We have shown that the game associated with our problem is convex, and characterized its core, which is interpreted as a converse for our setting. We have shown that the grand coalition is in the best interest of all agents and stable, in the sense that any coalition of agents has a disincentive  to leave the grand coalition. We have also characterized the Shapley value, and used it as a possible solution concept to ensure fairness among agents. 
 Finally, we have proposed an explicit coding scheme relying on polar codes for source coding and hash functions to achieve \emph{any point of the core}, including the Shapley value and the nucleolus. 
 
The framework is general and could be applied to other security problems involving a tension between cooperation and self-interest. The challenge
is in characterizing a value function for this framework. 
For instance, in our setting, being able to determine the value function in the non-degraded setting when $L>2$ remains an open problem, and is, unfortunately, at least as difficult as determining the secret-key capacity for the two-user secret generation model of~\cite{Maurer93}.

\appendices

\section{Proof of Lemma \ref{lem:lb}} \label{App_lem2}
Let $\mathcal{S} \subseteq \mathcal{L}$. We have
\begin{subequations}
\begin{align}
& H(X^N_{\mathcal{S}} | A_{\mathcal{L}}) \nonumber \\
& = H(X^N_{\mathcal{S}}  A_{\mathcal{S} } A_{\mathcal{L} \backslash \mathcal{S} } ) - H(A_{\mathcal{L}})  \\ 
&  = H(X^N_{\mathcal{S}}   A_{\mathcal{L} \backslash \mathcal{S} } ) - H(A_{\mathcal{L}})  \label{eq53a} \\ 
&  = H(X^N_{\mathcal{S}}) + H( A_{\mathcal{L} \backslash \mathcal{S} } |X^N_{\mathcal{S}}  ) - H(A_{\mathcal{L}}) \\ 
&  \geq H(X^N_{\mathcal{S}}) + H( A_{\mathcal{L} \backslash \mathcal{S} } |X^N_{\mathcal{S}}  ) - \log |\mathcal{A}_{\mathcal{L}}| \\ 
&  = H(X^N_{\mathcal{S}}) + H( A_{\mathcal{L} \backslash \mathcal{S} } |X^N_{\mathcal{S}}  ) - NH(X_{\mathcal{L}}|X_0) - o(N) \label{eq53b} \\ 
& = N I(X_{\mathcal{S}};X_0) + H( A_{\mathcal{L} \backslash \mathcal{S} } |X^N_{\mathcal{S}}  )  \nonumber\\
& \phantom{--} - NH(X_{\mathcal{L} \backslash \mathcal{S}}|X_0X_{\mathcal{S}}) - o(N), \label{eq2}
\end{align}
\end{subequations}
where \eqref{eq53a} holds because $A_{\mathcal{S} }$ is a function of $X^N_{\mathcal{S}}$, \eqref{eq53b} holds because 
\begin{subequations}
\begin{align}
\log |\mathcal{A}_{\mathcal{L}}|   
& = \sum_{l\in \mathcal{L}}\log |\mathcal{A}_{l}|\\ 
& = \sum_{l\in \mathcal{L}} |\mathcal{H}_{X_l|X_0X_{1:l-1}}|\\ 
& = \sum_{l\in \mathcal{L}} NH(X_l|X_0X_{1:l-1}) + o(N) \label{eq54a}\\
& = NH(X_{\mathcal{L}}|X_0) + o(N),
\end{align}
\end{subequations}
where \eqref{eq54a} holds by \cite[Theorem 3.5]{Sasoglu11,Arikan10,arikan2009rate}.

We lower-bound the second term in the right-hand side of~\eqref{eq2} as follows
\begin{subequations}
\begin{align}
&H( A_{\mathcal{L} \backslash \mathcal{S} } |X^N_{\mathcal{S}}  ) \nonumber \\
& \geq \sum_{j \in \mathcal{L} \backslash \mathcal{S} } H(A_j | A_{1:j-1}X^N_{\mathcal{S}}) \label{eq55a} \\  
& \geq \sum_{j \in \mathcal{L} \backslash \mathcal{S} } H(A_j | X_0^N X^N_{1:j-1}X^N_{\mathcal{S}})\\  
& = \sum_{j \in \mathcal{L} \backslash \mathcal{S} } H(U_j^N[\mathcal{H}_{X_j|X_0X_{1:j-1}}] |X_0^N X^N_{1:j-1}X^N_{\mathcal{S}} )\\ 
& \geq \sum_{j \in \mathcal{L} \backslash \mathcal{S} } H(U_j^N[\mathcal{V}_{X_j|X_0X_{1:j-1}X_{\mathcal{S}} }] |X_0^N X^N_{1:j-1}X^N_{\mathcal{S}} ) \label{eq55b} \\ 
&  \geq  \sum_{j \in \mathcal{L} \backslash \mathcal{S} } \sum_{ \substack{i \in \\ \mathcal{V}_{X_j|X_0X_{1:j-1}X_{\mathcal{S}} }} } \!\!\!\!\!\!\!\!\!\!\!\!H( (U_j)_i| (U_j)^{i-1} X_0^NX^N_{1:j-1}X^N_{\mathcal{S}} ) \label{eq55c} \\  
& \geq  \sum_{j \in \mathcal{L} \backslash \mathcal{S} } \sum_{ \substack{i \in \\ \mathcal{V}_{X_j|X_0X_{1:j-1}X_{\mathcal{S}} }} } (1- \delta_N) \label{eq55d}\\  
&=  \sum_{j \in \mathcal{L} \backslash \mathcal{S} } |\mathcal{V}_{X_j|X_0X_{1:j-1}X_{\mathcal{S}}}| (1- \delta_N)\\  
&  = \sum_{j \in \mathcal{L} \backslash \mathcal{S} } NH(X_j|X_0X_{1:j-1}X_{\mathcal{S}})  -o(N) \label{eq55e} \\  
&  = \sum_{j \in \mathcal{L} \backslash \mathcal{S} } NH(X_j|X_0X_{\llbracket 1:j-1 \rrbracket \cap \mathcal{L} \backslash \mathcal{S}}X_{\mathcal{S}})  -o(N)\\  
& = N H(X_{\mathcal{L} \backslash \mathcal{S}}|X_0X_{\mathcal{S}})-o(N) ,\label{eq1}
\end{align}
\end{subequations}
where \eqref{eq55a} and \eqref{eq55c} hold by the chain rule and because conditioning reduces entropy, \eqref{eq55b} holds because $\mathcal{H}_{X_j|X_0X_{1:j-1}} \supset \mathcal{H}_{X_j|X_0X_{1:j-1}X_{\mathcal{S}} } \supset \mathcal{V}_{X_j|X_0X_{1:j-1}X_{\mathcal{S}} }$ , \eqref{eq55d} holds by definition of $\mathcal{V}_{X_j|X_0X_{1:j-1}X_{\mathcal{S}}}$, \eqref{eq55e} holds by Lemma~\ref{lemscs2}, \eqref{eq1} holds by the chain rule.

Finally, combining \eqref{eq2} and \eqref{eq1} proves Lemma \ref{lem:lb}.

\section{Proof of Lemma \ref{lemloh}} \label{App_lem3}
We first prove the result when $Z=\emptyset$.  For $X_{\mathcal{L}},X_{\mathcal{L}}',F_{\mathcal{L}},F_{\mathcal{L}}'$  independent, we compute the following collision probability
\begin{subequations}
\begin{align}
&\mathbb{P} [ (F_{\mathcal{L}}(X_{\mathcal{L}}),F_{\mathcal{L}}) = (F_{\mathcal{L}}'(X_{\mathcal{L}}'),F_{\mathcal{L}}')]  \nonumber \\ 
&= \mathbb{P} [ F_{\mathcal{L}} = F_{\mathcal{L}}'] \mathbb{P} [ F_{\mathcal{L}}(X_{\mathcal{L}}) = F_{\mathcal{L}}(X_{\mathcal{L}}')]\\ 
&= \prod_{l \in \mathcal{L}} \mathbb{P} [ F_{l} = F_{l}'] \mathbb{P} [ F_{\mathcal{L}}(X_{\mathcal{L}}) = F_{\mathcal{L}}(X_{\mathcal{L}}')]\\ 
& = s^{-1}_{\mathcal{L}} \sum_{x_{\mathcal{L}},x_{\mathcal{L}}'} \mathbb{P} [ F_{\mathcal{L}}(x_{\mathcal{L}}) = F_{\mathcal{L}}(x_{\mathcal{L}}')] \mathbb{P} [X_{\mathcal{L}}=x_{\mathcal{L}}, X_{\mathcal{L}}'=x_{\mathcal{L}}']\\ 
& = s^{-1}_{\mathcal{L}} \sum_{x_{\mathcal{L}},x_{\mathcal{L}}'} \mathbb{P} [ F_{\mathcal{L}}(x_{\mathcal{L}}) = F_{\mathcal{L}}(x_{\mathcal{L}}')] \mathbb{P} [X_{\mathcal{L}}=x_{\mathcal{L}}] \mathbb{P} [X_{\mathcal{L}}'=x_{\mathcal{L}}']\\ 
& = s^{-1}_{\mathcal{L}} \sum_{\mathcal{S} \subseteq \mathcal{L}}\sum_{x_{\mathcal{L}}}  \smash{\sum_{\substack{x_{\mathcal{L}}' \\ \!\!\!\!\!\!\text{s.t.} x'_{\mathcal{S}} \neq x_{\mathcal{S}} \\  x'_{\mathcal{S}^c} = x_{\mathcal{S}^c} }}} \mathbb{P} [ F_{\mathcal{L}}(x_{\mathcal{L}}) = F_{\mathcal{L}}(x_{\mathcal{L}}')] \nonumber \\
&  \phantom{-----------} \times \mathbb{P} [X_{\mathcal{L}}=x_{\mathcal{L}}] \mathbb{P} [X_{\mathcal{L}}'=x_{\mathcal{L}}']\\ 
& = s^{-1}_{\mathcal{L}} \sum_{\mathcal{S} \subseteq \mathcal{L}}\sum_{x_{\mathcal{L}}} \smash{ \sum_{\substack{x_{\mathcal{L}}' \\ \!\!\!\!\!\!\text{s.t.} x'_{\mathcal{S}} \neq x_{\mathcal{S}} \\  x'_{\mathcal{S}^c} = x_{\mathcal{S}^c} }} \prod_{l\in \mathcal{L}} \mathbb{P} [ F_{l}(x_{l}) = F_{l}(x_{l}')] \mathbb{P} [X_{\mathcal{L}}=x_{\mathcal{L}}] } \nonumber \\
&  \phantom{------------} \times  \mathbb{P} [X_{\mathcal{S}}'=x_{\mathcal{S}}',X_{\mathcal{S}^c}'=x_{\mathcal{S}^c}]\\ 
& \leq s^{-1}_{\mathcal{L}}  \sum_{\mathcal{S} \subseteq \mathcal{L}}\sum_{x_{\mathcal{L}}} \smash{ \sum_{\substack{x_{\mathcal{L}}' \\ \!\!\!\!\!\!\text{s.t.} x'_{\mathcal{S}} \neq x_{\mathcal{S}} \\ x'_{\mathcal{S}^c} = x_{\mathcal{S}^c} }}}  2^{-r_{\mathcal{S}}} \mathbb{P} [X_{\mathcal{L}}=x_{\mathcal{L}}]\nonumber \\
&  \phantom{-----------} \times \mathbb{P} [X_{\mathcal{S}}'=x_{\mathcal{S}}',X_{\mathcal{S}^c}'=x_{\mathcal{S}^c}]  \label{eq56a}\\ 
& \leq s^{-1}_{\mathcal{L}}  \sum_{\mathcal{S} \subseteq \mathcal{L}}\sum_{x_{\mathcal{L}}}   2^{-r_{\mathcal{S}}} \mathbb{P} [X_{\mathcal{L}}=x_{\mathcal{L}}] \mathbb{P} [X_{\mathcal{S}^c}'=x_{\mathcal{S}^c}] \label{eq56b}\\ 
&\leq s^{-1}_{\mathcal{L}}  \sum_{\mathcal{S} \subseteq \mathcal{L}}\sum_{x_{\mathcal{L}}}   2^{-r_{\mathcal{S}}} \mathbb{P} [X_{\mathcal{L}}=x_{\mathcal{L}}] 2^{-H_{\infty}(p_{X_{\mathcal{S}^c}})} \label{eq56c} \\ 
& = s^{-1}_{\mathcal{L}}  \sum_{\mathcal{S} \subseteq \mathcal{L}}   2^{-r_{\mathcal{S}}-H_{\infty}(p_{X_{\mathcal{S}^c}})}, \label{eqloh}
\end{align}
\end{subequations}
where \eqref{eq56a} holds by the two-universality of the $F_l$'s, $l \in \mathcal{L}$, \eqref{eq56b} holds by marginalization over $X'_{\mathcal{S}}$, \eqref{eq56c}~holds by definition of the min-entropy.

Then, viewing $\mathbb{V} ( p_{F_{\mathcal{L}}(X_{\mathcal{L}}),F_{\mathcal{L}}}, p_{U_{\mathcal{K}}} p_{U_{\mathcal{F}}})$ as a scalar product between $(p_{F_{\mathcal{L}}(X_{\mathcal{L}}),F_{\mathcal{L}}} -  p_{U_{\mathcal{K}}} p_{U_{\mathcal{F}}})$ and its sign, by Cauchy-Schwarz inequality, we have  
\begin{subequations}
\begin{align}
& \mathbb{V} ( p_{F_{\mathcal{L}}(X_{\mathcal{L}}),F_{\mathcal{L}}}, p_{U_{\mathcal{K}}} p_{U_{\mathcal{F}}})^2   \nonumber \\
& \leq s_{\mathcal{L}} 2^{r_{\mathcal{L}}}  \sum_{m_{\mathcal{L}},f_{\mathcal{L}}} \left[ p_{F_{\mathcal{L}}(X_{\mathcal{L}}),F_{\mathcal{L}}} (m_{\mathcal{L}},f_{\mathcal{L}}) -  \frac{1}{s_{\mathcal{L}} 2^{r_{\mathcal{L}}} } \right]^2 \\ 
& = s_{\mathcal{L}} 2^{r_{\mathcal{L}}}  \left[ \sum_{m_{\mathcal{L}},f_{\mathcal{L}}}  p_{F_{\mathcal{L}}(X_{\mathcal{L}}),F_{\mathcal{L}}} (m_{\mathcal{L}},f_{\mathcal{L}})^2 \right] -  1  \\ 
&  = s_{\mathcal{L}} 2^{r_{\mathcal{L}}}  \mathbb{P} [ (F_{\mathcal{L}}(X_{\mathcal{L}}),F_{\mathcal{L}}) = (F_{\mathcal{L}}'(X_{\mathcal{L}}'),F_{\mathcal{L}}')] -   1 \\  
& \leq  2^{r_{\mathcal{L}}} \sum_{\mathcal{S} \subseteq \mathcal{L}}   2^{-r_{\mathcal{S}}-H_{\infty}(p_{X_{\mathcal{S}^c}})} -1 \label{eq57a}\\ 
& =   \sum_{\mathcal{S} \subsetneq \mathcal{L}}   2^{r_{\mathcal{S}^c}-H_{\infty}(p_{X_{\mathcal{S}^c}})} \\ 
& = \sum_{ \substack{ \mathcal{S} \subseteq {\mathcal{L}} \\ \mathcal{S}\neq \emptyset }} 2^{ r_{\mathcal{S}} - H_{\infty}\left( p_{X_{\mathcal{S}}} \right) }, \label{eqloh2}
\end{align}
\end{subequations}
where \eqref{eq57a} holds by \eqref{eqloh}.

We now introduce the random variable $Z$ correlated to $X_{\mathcal{L}}$ and proceed as in \cite{dodis2008fuzzy}. Let $z \in \mathcal{Z}$ and $X_{\mathcal{L}}^{(z)}$ be defined by $p_{X_{\mathcal{L}}^{(z)}} = p_{X_{\mathcal{L}}|Z=z}$. We have
\begin{subequations}
\begin{align}
&  \mathbb{V} ( p_{F_{\mathcal{L}}(X_{\mathcal{L}}),F_{\mathcal{L}},Z}, p_{U_{\mathcal{K}}} p_{U_{\mathcal{F}}}p_Z)  \nonumber\\
&=  \mathbb{E}_{Z} \left[ \mathbb{V} ( p_{F_{\mathcal{L}}(X^{(z)}_{\mathcal{L}}),F_{\mathcal{L}}}, p_{U_{\mathcal{K}}} p_{U_{\mathcal{F}}})  \right]  \\
& \leq  \mathbb{E}_{Z} \sqrt{ \sum_{ \substack{ \mathcal{S} \subseteq {\mathcal{L}} \\ \mathcal{S}\neq \emptyset }} 2^{ r_{\mathcal{S}} - {H}_{\infty}\left( {X_{\mathcal{S}}|Z=z} \right)} }\label{eq58a} \\
& \leq   \sqrt{ \sum_{ \substack{ \mathcal{S} \subseteq {\mathcal{L}} \\ \mathcal{S}\neq \emptyset }} 2^{ r_{\mathcal{S}} - {H}_{\infty}\left( {X_{\mathcal{S}}|Z} \right)}  \label{eq58b}
},
\end{align}
\end{subequations}
where \eqref{eq58a} holds by \eqref{eqloh2}, \eqref{eq58b} holds by Jensen's inequality  and since, by definition, $\mathbb{E}_{Z}[2^{ - {H}_{\infty}\left( {X_{\mathcal{S}}|Z=z}\right)}] = 2^{ - {H}_{\infty}\left( {X_{\mathcal{S}}|Z}\right)}$, $\mathcal{S} \subseteq \mathcal{L}$.

\section{Proof of Theorem \ref{sk2}} \label{App_sk2}
We consider the coding scheme of Section \ref{sec:CS1}. The only difference with the proof of Theorem \ref{sk} is that instead of Lemma \ref{lem:lb}, we now need to lower bound for any $\mathcal{T} \subseteq \mathcal{S}$, the quantity $H(X^N_{\mathcal{T}} | X^N_{{\mathcal{S}^c}}A_{\mathcal{S}})$. We do it as follows.
\begin{subequations}
\begin{align}
& H(X^N_{\mathcal{T}} | X^N_{{\mathcal{S}^c}}A_{\mathcal{S}}) \nonumber \\
& =H(X^N_{\mathcal{T}} X^N_{{\mathcal{S}^c}}A_{\mathcal{S} })  - H(X^N_{{\mathcal{S}^c}}A_{\mathcal{S}})  \\ 
& =H(X^N_{\mathcal{T}} X^N_{{\mathcal{S}^c}}A_{\mathcal{S} \backslash \mathcal{T} })  - H(X^N_{{\mathcal{S}^c}}A_{\mathcal{S}})  \\ 
& = H( X^N_{{\mathcal{T}}}|X^N_{\mathcal{S}^c}) + H(A_{\mathcal{S} \backslash \mathcal{T} }|X^N_{\mathcal{T}} X^N_{{\mathcal{S}^c}} ) \nonumber \\
& \phantom{--} -H(A_{\mathcal{S}}|X^N_{{\mathcal{S}^c}}). \label{eqapp}
\end{align}
\end{subequations}
We lower bound $H(A_{\mathcal{S} \backslash \mathcal{T} }|X^N_{\mathcal{T}}X^N_{{\mathcal{S}^c}})$ in the right hand side of~\eqref{eqapp} as follows.
\begin{subequations}
\begin{align} 
&H(A_{\mathcal{S} \backslash \mathcal{T} }|X^N_{\mathcal{T}}X^N_{{\mathcal{S}^c}}) \nonumber\\
& \geq N H(X_{\mathcal{S} \backslash \mathcal{T}}|X_0X_{\mathcal{T}}X_{{\mathcal{S}^c}})+o(N) \label{eq60a} \\
& = N H(X_{\mathcal{S} \backslash \mathcal{T}}|X_0)+o(N),  \label{eqapp2}
\end{align}
\end{subequations}
where \eqref{eq60a} holds similarly to \eqref{eq1} proved in Appendix~\ref{App_lem2} by conditioning on $X^N_{{\mathcal{S}^c}}$, \eqref{eqapp2} holds by the Markov chain \eqref{Markov}.
We then upper bound $H(A_{\mathcal{S}}|X^N_{{\mathcal{S}^c}})$ in the right hand side of~\eqref{eqapp} as follows.
\begin{subequations}
\begin{align}
& H(A_{\mathcal{S}}|X^N_{{\mathcal{S}^c}}) \nonumber \\
& \leq \log |\mathcal{A}_{\mathcal{S}}| \\ 
& = \sum_{i \in \mathcal{S}} \log |\mathcal{A}_{i}|\\ 
& = \sum_{i \in \mathcal{S}} |\mathcal{H}_{X_{i}|X_0X_{1:i-1}}| \\ 
& = \sum_{i \in \mathcal{S}}  N H(X_{i}|X_0X_{ 1:i-1})+o(N) \label{eq62a}\\ 
& = N  H( X_{\mathcal{S}} | X_0) +o(N) \label{eq62b} \\
& = N  H( X_{\mathcal{T}} | X_0) + N H(X_{\mathcal{S} \backslash \mathcal{T}} | X_0) +o(N) , \label{eqapp3}
\end{align}
\end{subequations}
where \eqref{eq62a} holds by \cite[Theorem 3.5]{Sasoglu11}, \eqref{eq62b} and \eqref{eqapp3} hold by the Markov chain \eqref{Markov}.
Hence, we obtain
\begin{subequations}
\begin{align}
& H(X^N_{\mathcal{T}} | X^N_{{\mathcal{S}^c}}A_{\mathcal{S}})  \nonumber \\
& \geq N H( X_{{\mathcal{T}}}|X_{\mathcal{S}^c})  - NH(X_{\mathcal{T}}|X_{0})-o(N)  \label{eq63a} \\
& = NI( X_{{\mathcal{T}}}; X_0|X_{\mathcal{S}^c}) -o(N), \label{eq63b}
\end{align}
\end{subequations}
where \eqref{eq63a} holds by combining \eqref{eqapp}, \eqref{eqapp2}, and \eqref{eqapp3}, \eqref{eq63b} holds by the Markov chain \eqref{Markov}.

\section{Proof of Corollary \ref{cor}} \label{App_cor}
We first show the following lemma and then use a similar argument than in \cite{madiman2008cores}. 
\begin{lem} \label{lemc}
Let $\mathcal{S} \subset \mathcal{L}$. The set function $w:2^{\mathcal{S}} \to \mathbb{R}^+, \mathcal{T} \mapsto I(X_{\mathcal{T}};X_0|X_{\mathcal{S}^c})$ is submodular, i.e., $-w$ is supermodular.	
\end{lem}
\begin{proof}
Let $\mathcal{U},\mathcal{V} \subseteq \mathcal{S}$. We have
\begin{subequations}
\begin{align}
	& I(X_{\mathcal{U}\cup\mathcal{V}};X_0|X_{\mathcal{S}^c}) + I(X_{\mathcal{U}\cap\mathcal{V}};X_0|X_{\mathcal{S}^c}) \nonumber \\
	& = I(X_{\mathcal{U}};X_0|X_{\mathcal{S}^c}) + I(X_{\mathcal{V}\backslash \mathcal{U}};X_0|X_{\mathcal{S}^c}X_{\mathcal{U}})\nonumber \\
& \phantom{--}+ I(X_{\mathcal{U}\cap\mathcal{V}};X_0|X_{\mathcal{S}^c})\\
	& = I(X_{\mathcal{U}};X_0|X_{\mathcal{S}^c}) + H(X_{\mathcal{V}\backslash \mathcal{U}}|X_{\mathcal{S}^c}X_{\mathcal{U}}) \nonumber \\
& \phantom{--}- H(X_{\mathcal{V}\backslash \mathcal{U}}|X_0X_{\mathcal{S}^c}X_{\mathcal{U}}) + I(X_{\mathcal{U}\cap\mathcal{V}};X_0|X_{\mathcal{S}^c})\\
	& \leq  I(X_{\mathcal{U}};X_0|X_{\mathcal{S}^c}) + I(X_{\mathcal{V}\backslash \mathcal{U}};X_0|X_{\mathcal{S}^c}X_{\mathcal{U}\cap \mathcal{V}})\nonumber \\
& \phantom{--}+ I(X_{\mathcal{U}\cap\mathcal{V}};X_0|X_{\mathcal{S}^c}) \label{eq64a}\\
	& = I(X_{\mathcal{U}};X_0|X_{\mathcal{S}^c}) + I(X_{\mathcal{V}};X_0|X_{\mathcal{S}^c}),
\end{align}
\end{subequations}
where \eqref{eq64a} holds because $H(X_{\mathcal{V}\backslash \mathcal{U}}|X_{\mathcal{S}^c}X_{\mathcal{U}}) \leq H(X_{\mathcal{V}\backslash \mathcal{U}}|X_{\mathcal{S}^c}X_{\mathcal{U} \cap \mathcal{V} })$ and because $H(X_{\mathcal{V}\backslash \mathcal{U}}|X_0X_{\mathcal{S}^c}X_{\mathcal{U}}) = H(X_{\mathcal{V}\backslash \mathcal{U}}|X_0X_{\mathcal{S}^c}X_{\mathcal{U}\cap \mathcal{V}})$ by the Markov chain \eqref{Markov}.
\end{proof}

For $\mathcal{S} \subset \mathcal{L}$, $(\mathcal{S},w)$ defines a concave game by submodularity of $w$ shown in Lemma \ref{lemc}. Consequently, its core $\mathcal{C}(w) \triangleq \left\{ (R_i)_{i \in \mathcal{S}} : \displaystyle\sum_{i \in \mathcal{S}} R_l = v(\mathcal{S}) \text{ and }\displaystyle\sum_{i \in \mathcal{T}} R_i \leq v(\mathcal{T}), \forall \mathcal{T} \subset \mathcal{S} \right\}
$ is non-empty by \cite{shapley1971cores}, i.e., there exists an achievable rate tuple in $\mathcal{R}_{\mathcal{S}}$ with sum-rate $I(X_{\mathcal{S}};X_0|X_{\mathcal{S}^c})$.

\bibliographystyle{IEEEtran}
\bibliography{polarwiretap}

\end{document}